\documentclass[11pt,reqno]{amsart}
\usepackage{amssymb}
\usepackage{enumerate}
\usepackage{hyperref}
\usepackage{doi,url}

\usepackage{color}

\definecolor{purple}{rgb}{.5,0,1}
\definecolor{orange}{rgb}{1,.5,0}
\definecolor{pink}{rgb}{1,0,.5}

\newcommand\cS{\mathcal{S}}

\numberwithin{equation}{section}
\newtheorem{theorem}{Theorem}[section]

\newtheorem{lemma}[theorem]{Lemma}

\newtheorem{remark}[theorem]{Remark}
\newtheorem*{assumption}{Property DL}
\newtheorem*{teks}{Droplet localization}
\newtheorem*{teks2}{Dynamical exponential clustering}

\DeclareMathOperator{\supp}{supp}
\DeclareMathOperator{\tr}{tr}
\DeclareMathOperator{\Ran}{Ran}

\DeclareMathOperator{\dist}{dist}

\DeclareMathOperator{\Rea}{Re}

\newcommand{\nn}{\notag}

\newcommand\R{\mathbb R}
\newcommand\N{\mathbb N}
\newcommand\C{\mathbb C}
\newcommand\Z{\mathbb Z}

\newcommand\e{\mathrm{e}}

\newcommand\E{\mathbb E}

\newcommand\cL{\mathcal{L}}

\newcommand\cN{\mathcal{N}}

\newcommand\cH{\mathcal{H}}

\newcommand\eps{\varepsilon}

\renewcommand{\d}{\mathrm{d}}

\newcommand{\pr}{\prime}

\newcommand{\bom}{{\boldsymbol{{\omega}}}}

\newcommand\beq{\begin{equation}}
\newcommand\eeq{\end{equation}}

\newcommand{\abs}[1]{\left\lvert #1 \right\rvert}
\newcommand{\norm}[1]{\left\lVert #1 \right\rVert}
\newcommand{\scal}[1]{\left\langle #1 \right\rangle}
\newcommand{\set}[1]{\left\{ #1 \right\}}
\newcommand{\pa}[1]{\left( #1 \right)}

\newcommand{\br}[1]{\left [ #1 \right]}

\newcommand\Ups{\Upsilon}

\newcommand{\eq}[1]{\eqref{#1}}

\newcommand{\up}[1]{^{(#1)}}

\newcommand{\qtx}[1]{\quad\text{#1}\quad}
\newcommand{\mqtx}[1]{\; \ \text{#1}\; \  }
\newcommand{\sqtx}[1]{\;\text{#1}\;}

\newcommand{\lqtx}[1]{\quad\text{#1}}

\begin{document}

\title[Dynamical localization in the disordered XXZ  spin chain]{Manifestations of dynamical localization in the disordered XXZ  spin chain}

\author{Alexander Elgart}
\address[A. Elgart]{Department of Mathematics; Virginia Tech; Blacksburg, VA, 24061, USA}
 \email{aelgart@vt.edu}

\author{Abel Klein}
\address[A. Klein]{University of California, Irvine;
Department of Mathematics;
Irvine, CA 92697-3875,  USA}
 \email{aklein@uci.edu}

\author{G\"unter Stolz}
\address[G. Stolz]{Department of Mathematics; University of Alabama at Birmingham;
Birmingham, AL 35294-1170, USA}
\email{stolz@uab.edu}

\thanks{A.K. was  supported in part by the NSF under grant DMS-1001509.}

\date{Version of \today}

\begin{abstract}

We study disordered XXZ spin  chains in the Ising phase exhibiting droplet localization, a single cluster localization property  we previously proved for  random XXZ spin chains. It holds in  an energy interval $I$ near the bottom of the spectrum, known as the droplet spectrum. We establish dynamical manifestations of localization  in the energy window $I$, including non-spreading of information,  zero-velocity Lieb-Robinson bounds, and general dynamical clustering. Our results do not rely on  knowledge of the dynamical characteristics of the model outside the droplet spectrum. 
 A byproduct of our analysis is that   for random  XXZ spin  chains this droplet localization can happen only inside the droplet spectrum.
\end{abstract}

\maketitle

\tableofcontents

\section{Introduction}

We study disordered XXZ spin  chains in the Ising phase exhibiting droplet localization.
This is a single cluster localization property  we previously proved for  random XXZ spin chains inside the droplet spectrum \cite{EKS}.

The basic phenomenon of Anderson localization  in the single particle framework is that disorder can cause localization of electron states and thereby manifest itself in properties such as non-spreading of wave packets under time evolution and absence of dc transport. The mechanism behind this behavior is  well understood  by now, both physically and mathematically (e.g., \cite{And,EM,Kirsch, GKber,AW,EK}).   Many manifestations of single-particle Anderson localization 
remain valid  if one considers a fixed number of interacting 
particles, e.g., \cite{CS,AW2,KN}.

The situation is radically different in the many-body setting. Little is known about the thermodynamic limit of an interacting electron gas in a random environment, i.e., an infinite volume limit in which the number of electrons  grows proportionally to the volume. Even  simplest models where the individual particle Hilbert space is  finite dimensional (spin systems) pose considerable analytical and numerical challenges, due to the fact that the number of degrees of freedom involved grows exponentially fast with the  size of the system. 

The limited evidence from perturbative  \cite{AF,AGKL,GMP,AAB,PPZ,Imb} and  numerical \cite{GM,BR,HO,HP}  approaches supports the persistence of a many-body localized (MBL) phase for one-dimen\-sional spin systems  in the presence of weak interactions.  The numerics also suggests the existence of transition from a many-body localized (MBL) phase to delocalized phases as the strength of interactions increases,  \cite{HO,HP,BPM,BN,SPA}. 	

Mathematically rigorous results on localization in a true many-body system have been until very recently confined to investigations of exactly solvable (quasi-free) models (see \cite{KP,ARNSS,SW}). More recent  progress has been achieved primarily in the study of the XXZ spin chain, a system that is not integrable but yet amenable to rigorous analysis. The  first results in this direction established the exponential clustering property for zero temperature correlations of the Andr\'e-Aubry quasi-periodic model  \cite{Mas1,Mas2}. The authors recently proved localization results for   the random XXZ spin chain in the droplet spectrum \cite{EKS}.  Related results are given in \cite{BW}.

In  {\cite[Theorem~2.1]{EKS}},  the authors 
 obtained a strong localization result   for the  droplet spectrum eigenstates  of  the random XXZ spin chain  in the Ising phase. This result can be interpreted as the statement that a typical eigenstate in this part of the spectrum behaves as an effective quasi-particle,  localized, in the appropriate sense, in the presence of a random field.
 
In this paper we study  disordered XXZ spin  chains  exhibiting the same localization property we proved in {\cite[Theorem~2.1]{EKS}}, which we call Property DL (for ``droplet localization").
 We  draw  conclusions concerning the dynamics of the  spin  chain  based exclusively on
Property DL.

For completely localized many-body systems, the dynamical manifestation of localization is often expressed in terms of  the non-spreading  of information under the time evolution.   An alternative (and equivalent) description is  the zero-velocity Lieb-Robinson  bound.  (See, e.g, \cite{FHBSE}.)

There is, however, 
a difficulty in even formulating our results for  disordered XXZ spin  chains.  Property DL  only carries information 
about the structure of the eigenstates near the bottom of the spectrum,  and we cannot assume complete localization for all energies. Moreover, Theorem~\ref{thmrigid} below shows that Property DL can only hold inside the droplet spectrum for random XXZ spin chains, showing the near optimality of the interval in  [16, Theorem 2.1].
In fact, numerical studies suggest the presence of a mobility edge  for  sufficiently small disorder, \cite{HO,HP,BPM,BN}. To resolve this issue, we recast  non-spreading  of information and the zero-velocity Lieb-Robinson bound as a problem on the subspace of the Hilbert space associated with the given energy window in which Property DL holds. This leads to a number of interesting findings, formulated below in
Theorem~\ref{corquasloc}  (non-spreading of information), Theorem~\ref{thm:expclusteringgenJ}  
(zero-velocity Lieb-Robinson bounds), and Theorem~\ref{thmexpclust} (general dynamical clustering).

As we mentioned earlier, our methodology in \cite{EKS}  is limited to the states near the bottom of the spectrum and sheds light only on what physicists call zero temperature localization.  It is unrealistic to expect that this approach  can yield insight about extensive energies of magnitude comparable to the system size which is the essence of MBL. Nonetheless, we believe that the ideas presented here will be useful in understanding the transport properties of  interacting systems that have a mobility edge, such as the Quantum Hall Effect \cite{GP, EGS,GKS}.

Some of the results in this paper were announced in \cite{EKS2}.

This paper is organized as follows:  The model, Property DL, and the main theorems are stated in Section~\ref{secmodel}. We collect some technical results in Section~\ref{secprel},  and a lemma about spin chains is presented in  Appendix~\ref{appspinc}.
  Section~\ref{secopt} is devoted to the proof that Property DL only holds inside the droplet spectrum for random XXZ spin chains (Theorem~\ref{thmrigid}).  Non-spreading of information (Theorem~\ref{corquasloc}) is proven in Section~\ref{secnonsp}. Zero-velocity Lieb-Robinson bounds (Theorem~\ref{thm:expclusteringgenJ}) are proven in Section~\ref{secLR}. Finally, the proof of general dynamical clustering (Theorem~\ref{thmexpclust}) is given in Section~\ref{secdyncl}.

\section{Model and results}\label{secmodel}

The infinite disordered XXZ spin  chain  (in the Ising phase) is given by the (formal) Hamiltonian 
\beq \label{infXXZ}
H=H_\omega=H_0+ \lambda B_\omega,\quad H_0=\sum_{i\in\Z}h_{i,i+1},\quad B_\omega=\sum_{i\in\Z} \omega_i \mathcal{N}_i,
\eeq
acting on   $\bigotimes_{i\in \Z} \C_i^2$,  with $\C_i^2=\C^2$ for all $i\in \Z$, the quantum spin configurations on the one-dimensional lattice $\Z$, where  
\begin{enumerate}
\item $h_{i,i+1}$, the local next-neighbor Hamiltonian, is given by 
\beq 
h_{i,i+1}=\tfrac{1}{4}\pa{I-\sigma_i^z\sigma_{i+1}^z}-\tfrac{1}{4\Delta}\pa{\sigma_i^x\sigma_{i+1}^x+\sigma_i^y\sigma_{i+1}^y},
\eeq
where $\sigma^{x},\sigma^{y},\sigma^{z}$ are the standard Pauli matrices ($\sigma_i^{x},\sigma_i^{y},\sigma_i^{z}$  act on $\C_i^2$)  and $\Delta>1$ is a  parameter;
\item 
$\mathcal{N}_i = \tfrac{1}{2} (1-\sigma_i^z)
$
is the local number operator at site $i$
(the projection onto the down-spin state at site $i$); 

\item $\omega = \set{\omega_i}_{i\in\Z}$ are identically
distributed random variables whose joint probability distribution is ergodic with respect to shifts in $\Z$, and the single-site 
   probability distribution   $\mu$ satisfies  
   \beq
    \set{0,1}\subset \supp \mu\subset[0,1] \qtx{and}\mu(\set{0})=0;
    \eeq  
   
    \item  $\lambda > 0$ is the disorder parameter.
   \end{enumerate}
 If in addition   $\set{\omega_i}_{i\in\Z}$ are  independent  random variables we call $H_\omega$  a \emph{random} XXZ spin chain.
 \medskip

The choice $\Delta>1$  specifies the Ising phase.  The Heisenberg chain corresponds to $\Delta=1$, and  the Ising chain is obtained  in the limit $\Delta\to\infty$.

We set  $e_+ = \begin{pmatrix} 1  \\ 0  \end{pmatrix} $ and  $e_- = \begin{pmatrix} 0  \\ 1 \end{pmatrix} $, spin up  and spin down, respectively. 
 Recall $\sigma^{z} e_{\pm}= \pm  e_{\pm}$.  Thus, if $\mathcal{N} = \tfrac{1}{2} (1-\sigma^z)
$, we have $\mathcal{N} e_+=0$ and $\mathcal{N} e_-=e_-$.

The operator $H_\omega$ as in \eq{infXXZ} with $B_\omega \ge 0$ can be defined as an unbounded nonnegative self-adjoint operator as follows: 
Let $\cH_0$ be the vector subspace of $\bigotimes_{i\in \Z} \C_i^2$ spanned by tensor products of the form  $\bigotimes_{i\in \Z}e_i$,  $e_i\in \set{e_+,e_-}$, with a finite number of spin downs, equipped with   the tensor product inner product, and let $\cH$ be its Hilbert space completion. $H_\omega$,  defined in $\cH_0$ by \eq{infXXZ},  is an essentially self-adjoint operator on $\cH$.   Moreover, the ground state energy of $H_\omega$ is $0$, with the unique  ground state (or \emph{vacuum}) given by the all-spins up configuration $\psi_0= \otimes_{i\in \Z} e_{+}$.   Note that  $\cN_i \psi_0=0$ for all $i\in \Z$ and  $\norm{\psi_0}=1$.

The spectrum of $H_0$ is known to be of the form   \cite{NSt,FS} (recall $\Delta >1$):
\beq
\sigma(H_0)=\set{0} \cup \br{1 -\tfrac 1 \Delta, 1 +\tfrac 1 \Delta}  \cup \set{\left[2\pa{1 -\tfrac 1 \Delta},\infty\right )\cap\sigma(H_0) }.
\eeq
We will call  $I_1 =[1-\frac{1}{\Delta}, 2(1-\frac{1}{\Delta}))$   the \emph{droplet spectrum}.
(Droplet states in the Ising phase of the XXZ chain were first described in \cite{NSt} (see also \cite{NSS,FS}); they have energies in the interval $ \br{1 -\tfrac 1 \Delta, 1 +\tfrac 1 \Delta}$.
The pure droplet spectrum is actually $I_1\cap \sigma(H_0)$;  we call $I_1$ the droplet spectrum for convenience.)

 Since the  disordered XXZ spin chain Hamiltonian  $H_\omega$  is  ergodic with respect  to translation in $\Z$,  $B_\omega\ge 0$, and $B_\omega \psi_0=0$,  standard considerations imply that $H_\omega$ has nonrandom spectrum $\Sigma$, and 
\beq
\sigma(H_\omega)= \Sigma= \set{0} \cup \set{\left[1 -\tfrac 1 \Delta,\infty \right ) \cap \Sigma}\lqtx{almost surely}.
\eeq
(In the case of a  random XXZ spin chain Hamiltonian  $H_\omega$ with a continuous single-site probability distribution  standard arguments yield
$\Sigma=  \set{0} \cup \left[1 -\tfrac 1 \Delta,\infty \right )$.)

We consider the restrictions of $H_\omega$ to finite intervals $[-L,L]$, $L\in \N$ (We will write $[-L,L]$ for $[-L,L]\cap \Z$, etc., when it is clear from the context.)  We let  $\cH\up{L}=\cH_{[-L,L]}$,
where $\cH_S =\otimes_{i\in S}\C_i^2$ for $S\subset \Z$ finite,  and define  the self-adjoint operator  
\beq \label{finiteXXZ}
H^{(L)} =H_\omega ^{(L)}= \sum_{i=-L}^{L-1} h_{i,i+1} + \lambda\sum_{i=-L}^L \omega_i \mathcal{N}_i + \beta (\mathcal{N}_{-L} + \mathcal{N}_L) \qtx{on}\cH\up{L}.
\eeq
We take (and fix)  $\beta \ge \frac{1}{2}(1-\frac{1}{\Delta})$ in the boundary term, 
 which guarantees that  the random spectrum of $H_\omega ^{(L)}$ preserves the spectral gap of size $1 -\frac 1 \Delta$ above the ground state energy:
 \beq\label{gapcond}
 \sigma(H^{(L)}_\omega)= \set{0} \cup \set{\left[1 -\tfrac 1 \Delta,\infty \right ) \cap  \sigma(H^{(L)}_\omega)}.
 \eeq
The ground state energy of $H_\omega ^{(L)}$ is $0$,  with the all-spins up configuration state $\psi_0 ^{(L)}=\otimes_{i\in [-L,L]} \e_+ \in \cH\up{L}$ being a ground state, which is unique almost surely since $\sum_{i=-L}^L \omega_i \mathcal{N}_i \ne 0$ almost surely (which rules out  the all-spins down configuration in $H_\omega ^{(L)}$ as a ground state).

Given an interval $I$, we set $\sigma_I(H_\omega^{(L)})=\sigma(H_\omega^{(L)})\cap I$,  and  let  
\beq
G_I= \set{g:\R \to \C \mqtx{Borel measurable,}  \abs{g}\le  \chi_I}.
\eeq

In this article we consider a disordered XXZ spin  chain  as in \eq{infXXZ} for which we have localization in an interval  $\br{1 -\tfrac 1 \Delta, \Theta_1}$ in the following form, where $\norm{ \ }_1$ is the trace norm. 

\begin{assumption}   Let $H=H_\omega$ be a  disordered XXZ spin  chain.  There exist  $\Theta_1 > \Theta_0= 1 -\tfrac 1 \Delta$ and constants 
$C<\infty$ and $m>0$, such that, setting $I= [\Theta_0,\Theta_1]$,  we have, uniformly in $L$, 
\beq \label{eq:efcorDL}
 \E\pa{\sup_{g \in G_{ I} }\norm{\mathcal{N}_i g(H^{(L)}) \mathcal{N}_j}_1} \le C e^{-m|i-j|} \mqtx{for all}  i, j \in[-L, L].
\eeq
\end{assumption}

This property is justified because  we have proven its validity in the droplet spectrum \cite{EKS} for \emph{random}
XXZ spin  chains.  The name Property DL (for Droplet Localization) is further justified by 
Theorem~\ref{thmrigid} below.

  If $H=H_\omega$ is a random
XXZ spin  chain, then
 $H^{(L)}$ almost surely has simple spectrum. A simple analyticity based argument for this can be found in \cite[Appendix~A]{ARS}.  (The argument is presented there for the XY chain, but it holds for every random spin chain of the form $H_0 + \sum_{k=-L}^L \omega_k \mathcal{N}_k$ in $\bigotimes_{i\in [-L,L]}\C_i^2$.) Thus, almost surely, all its normalized eigenstates can be labeled as $\psi_E$ where $E$ is the corresponding eigenvalue.  In particular,
 \beq\label{norm1}
\norm{\mathcal{N}_i P_E^{(L)} \mathcal{N}_j}_1= \norm{\mathcal{N}_i\psi_E}\norm{\mathcal{N}_j\psi_E},
\eeq
where $P_E^{(L)}=\chi_{\set{E}}(H^{(L)})$ and $\norm{ \ }_1$ is the trace norm.

 Given  $0\le \delta< 1$, we  set
\beq \label{dropspec}
I_{1,\delta} = \left[ 1- \tfrac{1}{\Delta}, (2-\delta)\big(1-\tfrac{1}{\Delta}\big) \right];
\eeq
 note that  $I_{1,\delta}\subsetneq  I_1$ if $0<\delta <1$.
The following result  is proved in \cite{EKS}.

\begin{teks}[{\cite[Theorem~2.1]{EKS}}] 
 Let $H=H_\omega$ be a random XXZ spin chain whose single-site probability distribution is absolutely continuous with a bounded density.
There exists a constant $K>0$ with the following property: If $\Delta >1$, $\lambda >0$, and  $0<\delta< 1$ satisfy
\beq\label{lambdaDeltahyp}
 \lambda \pa{\delta(\Delta -1)}^{{\frac 1 2} } \min \set{1, \pa{\delta(\Delta -1)}}\ge K ,
 \eeq 
then  there exist constants 
$C<\infty$ and $m>0$ such that we have, uniformly in $L$, 
\beq \label{eq:efcor5}
 \E\pa{ \sum_{E\in \sigma_{ I_{1,\delta}}(H^{(L)})} \norm{\mathcal{N}_i\psi_E}\norm{\mathcal{N}_j\psi_E}} \le C e^{-m|i-j|} \mqtx{for all}  i, j \in[-L, L],
\eeq
and, as a consequence,
\beq \label{eq:efcor59}
 \E\pa{\sup_{g \in G_{ I_{1,\delta}}} \norm{\mathcal{N}_i g(H^{(L)}) \mathcal{N}_j}_1} \le C e^{-m|i-j|} \mqtx{for all}  i, j \in[-L, L].
\eeq
\end{teks}

The interval $I_{1,\delta}$ in \cite[Theorem~1.1]{EKS} is close to optimal, as the following theorem  shows that for a random XXZ spin chain  localization as in  \eq{eq:efcorDL} is only allowed in the droplet spectrum.

\begin{theorem}[Optimality of the droplet spectrum]\label{thmrigid}  Suppose Property DL  is valid for a random XXZ spin chain $H$.  Then
$\Theta_1\le 2 \Theta_0$, that is, if  $I$ is the interval in Property DL,  then we must have 
$I= I_{1,\delta}$ for some $0\le \delta <1$.
\end{theorem}

\emph{Let   $H_\omega$  be a disordered  XXZ spin chain satisfying Property DL.}
We consider the intervals $I=[\Theta_0,\Theta_1]$ and $I_0= [0,\Theta_1]$, where  $\Theta_0,\Theta_1$ are given in  Property DL. We mostly omit $\omega$ from the notation. We write  $P^{(L)}_B=\chi_B(H^{(L)})$ for a Borel set
$B\subset \R$, and let $P^{(L)}_E=P^{(L)}_{\set{E}}$ for $E\in \R$. It follows from \eq{gapcond} that 
$P^{(L)}_{I_0}= P^{(L)}_0 + P^{(L)}_{I}$.  Since $\cN_i P^{(L)}_0=P^{(L)}_0\cN_i =0$ for all $i \in [-L.L]$, $G_I$ may be replaced by  $G_{I_0}$ in  \eq{eq:efcorDL}.  By   $m>0$ we will always  denote the constant in  \eq{eq:efcorDL}.    $C$ will always denote a constant, independent of the relevant parameters, which may vary from equation to equation, and even inside the same equation.

Given an interval  $J \subset [-L,L]$, a local observable $X$  with support $J$ is an operator on $\otimes_{j\in J} \C_j^2$, considered as an operator on $\cH^{(L)}$ by acting as the identity on spins not in $J$. (We defined supports as intervals for convenience.  Note  that we do not ask $J$ to be the smallest interval with this property, supports of observables are not uniquely defined.)  

Given a local observable $X$, we will generally  specify   a support for $X$, denoted by $\cS_X=[s_X,r_X] $.    We always assume $\emptyset \ne \cS_X  \subset [-L,L]$.   Given two local observables $X, Y$ we set $\dist(X,Y)= \dist (\cS_X,\cS_Y)$.

Given $\ell \ge 1$ and $B\subset [-L,L]$, we set $B_\ell = \set{j \in  [-L,L]; \ \dist\pa{j,B} }\le \ell$.
In particular, given a local observable $X$ we let
\beq\label{cSell}
\cS_{X,\ell}=  \pa{\cS_X}_\ell =[s_X-\ell,r_X +\ell] \cap [-L,L].
 \eeq

In this paper we derive several manifestations of dynamical localization for $H$  from Property DL.  The time evolution of a local observable under $ H^{(L)}$ is given by  
\beq
\tau^{(L)}_t\pa{X}=\e^{itH^{(L)}}X\e^{-itH^{(L)}} \qtx{for} t\in \R.
\eeq
(We also  mostly  omit  $L$ from the notation, and write $\tau_t$ for $\tau^{(L)}_t$. ) 

For a completely localized many-body system (i.e., localized at all energies), dynamical localization is often expressed as  the \emph{non-spreading of information under the time evolution}:  Given a local observable $X$, for all $\ell \ge 1$ and $t\in \R$ there is a local observable $X_\ell(t)$ with support $\cS_{X,\ell}$, such that  $\norm{X_\ell(t)- \tau_t\pa{X}} \le C\norm{X} \e^{-c \ell}$, with the constants $C$ and $c>0$ independent of $X$, $t$, and $L$.  Since we only have localization in the energy interval $I$, and hence also in $I_0$, we should only expect non-spreading of information  in these energy intervals.

 Thus, given an energy interval $J$, we consider the sub-Hilbert space  
$\cH\up{L}_ J=\Ran P^{(L)}_J$, spanned by the the eigenstates of $ H\up{L}$  with energies in $ J$, and localize an observable  $X$ in the  energy interval $J$ by considering
 its restriction to  $\cH\up{L}_ J$, 
$X_J= P_J\up{L} X P_J\up{L}$.  Clearly $\tau_t \pa{X_J\up{L}}= \pa{\tau_t \pa{X\up{L}}}_J$.

Property DL implies non-spreading of information in the energy interval $I_0$.

\begin{theorem}[Non-spreading of information]\label{corquasloc} Let  $H=H_\omega$  be a disordered  XXZ spin chain satisfying Property DL.
 There exists $C<\infty$, independent of $L$, such that for all local observables $X$, 
 $t\in \R$ and  $\ell >0$ there is a local observable  $X_\ell(t)=\pa{X_\ell(t)}_\omega $  with support $\cS_{X,\ell}$   satisfying
\begin{align} \label{qlocI}
\E \pa{\sup_{t\in \R}\norm{  \pa{X_\ell(t)  - \tau_t\pa{ X }}_{I_0}}_1} \le C  \|X\|\e^{- \frac{1}{16} m\ell}.
\end{align}
\end{theorem}

We  give an explicit expression for $X_\ell(t)$ in \eq{eq:X_ell}.  
Note that   $  X_I= \pa{ X_{I_0} }_I$, and hence \eq{qlocI} implies the same statement with $I$ substituted for $I_0$.

Another manifestation of dynamical localization is the existence of zero-velocity Lieb-Robinson (LR) bounds \emph{in the interval of localization}. The following theorem states a zero-velocity Lieb-Robinson  bound in the energy  interval $I$.
    If we include the ground state, i.e., if  we look for Lieb-Robinson type bounds in  the energy  interval $I_0$, the situation is more complicated, and the zero-velocity Lieb-Robinson  bound holds for the double commutator;   the commutator requires counterterms.   Note that 
    $ [\tau_t\pa{ X_I },Y_I] \ne \pa{ [\tau_t\pa{ X },Y]}_I$.  (We mostly omit   $\omega$ and $L$  from the notation.)

\begin{theorem}[Zero velocity LR bounds]\label{thm:expclusteringgenJ}  Let  $H=H_\omega$  be a disordered  XXZ spin chain satisfying Property DL. Let $X, Y$ and $Z$ be local observables.   The following holds  uniformly in $L$:

\begin{align}\label{eq:dynloc}
& \E \pa{\sup_{t\in \R} \norm{ [\tau_t\pa{ X_I },Y_I]}_1} \le  C \|X\| \|Y\| \e^{-\frac 1 8 m\dist (X,Y)},\\
\label{eq:LRquasimix}
&\E \pa{\sup_{t\in \R}\norm{\left[ \tau_t\pa{X_{I_0}},Y_{I_0}\right]- \pa{\tau_t\pa{ X}P_0Y - YP_0 \tau_t\pa{ X }}_I}_1} \\   & \hskip114pt \le C \|X\| \|Y\| \e^{-\frac 1 8m\dist (X,Y)},  \nn  \\
 \label{eq:LRquasimix2} 
&\E\pa{\sup_{t,s \in \R} {\norm{\left[\left[ \tau_t\pa{X_{I_0}}, \tau_s\pa{Y_{I_0}}\right], Z_{I_0}\right] }}_1} \\  \nn
& \hskip80pt \le C  \|X\| \|Y\|\|Z\|\e^{-\frac 1 8m \min\set{\dist (X,Y), \dist (X,Z), \dist (Y,Z)}}.
\end{align} 
Moreover,  for  the random   XXZ spin chain  the estimate \eq{eq:LRquasimix}    is not true  without the counterterms.
\end{theorem}

  The counterterms in  \eq{eq:LRquasimix}   are generated by the interaction between the ground state and states corresponding to the energy interval $I$ under the dynamics.  Here, and also in Theorem~\ref{thmexpclust} below, they are linear combinations of terms of the form $ \pa{\tau_t\pa{ X}P_0Y}_I$ and  $\pa{YP_0 \tau_t\pa{ X }}_I$. 
 Note that 
\begin{align}\nn
\norm{\pa{\tau_t\pa{ X}P_0Y}_I}_1& =\norm{\pa{\tau_t\pa{ X}P_0Y}_I} =\norm{P_I Y^* \psi_0}\norm{P_I X \psi_0},\\
\norm{\pa{YP_0 \tau_t\pa{ X }}_I}_1& =\norm{\pa{YP_0 \tau_t\pa{ X }}_I} =\norm{P_I X^* \psi_0}\norm{P_I Y \psi_0},
\end{align}
which do not depend  on  either  $t$ or $\dist (X,Y)$.

Another manifestation of  localization is the dynamical exponential clustering property. Let   $B\subset \R$ be a Borel set.  We define   the  truncated time evolution   of an observable $X$  by  ($H=H_\omega^{(L)}$), 
\beq
\tau^B_t\pa{X}=\e^{itH_B}X\e^{-itH_B}, \qtx{where} H_B=P_B H.
\eeq
Note that  $\pa{\tau^B_t\pa{X}}_B= \pa{\tau_t\pa{X}}_B=\tau_t\pa{X_B}$.

The correlator operator  of two observables  $X$ and $Y$ in the energy window  $B$ is given by ($\bar P_B=1- P_B$)
\beq
R_{B}  (X,Y)=  P_B X \bar P_B Y P_B = \pa{X \bar P_B Y }_B.
\eeq
If $E$ is a simple eigenvalue with normalized eigenvector $\psi_E$, we have, with  $R_{E}  (X,Y) =R_{\set{E}}  (X,Y)$,
\begin{align}
{\tr \pa{R_{E}  (X,Y)} }&={\scal{\psi_E,XY\psi_E}-\scal{\psi_E,X\psi_E}\scal{\psi_E,Y\psi_E}}.
\end{align}

 The following result  is proved in \cite{EKS}.

\begin{teks2}[{\cite[Theorem~1.1]{EKS}}]Let $H=H_\omega$ be a random XXZ spin chain,
and assume \eq{eq:efcor5} holds in an interval $I$.  Then, for all local observables  $X$ and $Y$ we have, uniformly in $L$, 
\beq \label{eq:expclustering}
\E \pa{ \sup_{t\in \R}  \sum_{E\in \sigma_I(H^{(L)}) } \abs{\tr \pa{R_{E}  (\tau_t^I\pa{X},Y)} } }\le C \|X\| \|Y\| \e^{-m \dist\pa{ X, Y}},
\eeq
\beq \label{eq:expclustering2}
\E \pa{ \sup_{t\in \R}  \sum_{E\in \sigma_I(H^{(L)}) }    \abs{\tr \pa{R_{E}  (\tau_t\pa{X_I},Y_I)} } }\le C \|X\| \|Y\| \e^{-m \dist\pa{ X, Y}},
\eeq
and
\beq \label{eq:expclustering3}
\E \pa{ \sup_{t\in \R}    \abs{\tr \pa{R_{I}  (\tau_t^I\pa{X},Y)} } }\le C \|X\| \|Y\| \e^{-m \dist\pa{ X, Y}}.
\eeq
\end{teks2}

 The estimate \eq{eq:expclustering2} is not the same as \eq{eq:expclustering}, but it can be proven the same way; the proof of  \cite[Lemma~3.1]{EKS}   is actually simpler in this case. 
 
Since
 \beq
 \tr \pa{R_{I}  (\tau_t^I\pa{X},Y)}  =   \sum_{E\in \sigma_I(H^{(L)}) } \scal{\psi_E,\tau_t^I\pa{X}  \bar P_{I}  Y \psi_E},
 \eeq
 \eq{eq:expclustering3} is a statement about the  diagonal elements of the correlator operator 
 $R_{I}  (\tau_t^I\pa{X},Y) $.  We will now state a more general dynamical   clustering
result that is not restricted to diagonal elements.  The result, which holds in an interval of localization satisfying the conclusions of Theorem~\ref{thmrigid},  requires counterterms.

\begin{theorem}[General dynamical clustering]  \label{thmexpclust} Let  $H=H_\omega$  be a disordered  XXZ spin chain satisfying Property DL.  Fix an interval
$K= [\Theta_0, \Theta_2]$, where  $ \Theta_0 < \Theta_2 <\min\set{2 \Theta_0, \Theta_1 }$,  and  $\alpha\in(0,1)$.  There exists  $\tilde m>0$, such that
  for all  local observables  $X$ and $Y$ we have, uniformly in $L$,
\begin{multline}\label{eq:expclusteringgen'}
\E \pa{\sup_{t\in \R}\norm{R_K\pa{ \tau^K_t\pa{ X },Y} - \pa{\tau^K_t(X)P_0 Y +\tau^K_t\pa{Y} P_0 X }_K}}\\   C\pa{1 + \ln\pa{\min\set{\abs{\cS_{X}},\abs{\cS_{Y}}}}} \|X\| \|Y\|  \e^{-\tilde m \pa{\dist (X,Y)}^\alpha},
\end{multline}
and
\begin{align}\label{eq:dynloc3333}
&\E \pa{\sup_{t\in \R} \norm{ \pa{ [[\tau^K_t\pa{ X },Y]]}_K }}\\  \nn & \hskip30pt \le  C\pa{1 + \ln\pa{\min\set{\abs{\cS_{X}},\abs{\cS_{Y}}}}} \|X\| \|Y\|  \e^{-\tilde m \pa{\dist (X,Y)}^\alpha},
\end{align}
where 
\begin{align}\label{[[]]}
&[[\tau^K_t\pa{ X },Y]]= 
 [\tau^K_t\pa{ X },Y] \\& \nn \hskip30pt - \pa{\tau^K_t\pa{ X }P_0Y +\tau^K_t\pa{Y} P_0 X} + \pa{ Y P_0\tau^K_t\pa{X}  +  X P_0\tau^K_t\pa{Y} } .
\end{align}
Moreover,  for  the random   XXZ spin chain  the estimates \eq{eq:expclusteringgen'} and \eq{eq:dynloc3333}  are not true  without the counterterms.
\end{theorem}

While  it is  obvious where the counterterms in \eq{eq:LRquasimix} come from, the same is not
true in \eq{eq:expclusteringgen'}, where the time evolution in the second term seems to sit in the \emph{wrong} place: it is $\tau^K_t\pa{Y}$ and not $\tau^K_t\pa{X}$. It turns out this term encodes  information about the states above the energy window $K$, and  the appearance of $\tau^K_t\pa{Y}$ is related to the reduction of this data to $P_0$, as can be seen in the proof.

\begin{remark}\label{remwonder}
One may  wonder why the counterterms  in \eq{eq:expclusteringgen'} do not appear in \eq{eq:expclustering3}.  The reason is that their traces obey decay estimates similar to  \eq{eq:expclustering3}  with $\alpha=1$, see Lemma~\ref{lemwonder}. 
 \end{remark}

\section{Preliminaries}\label{secprel}

\subsection{Decomposition of  local observables}

Given $S\subset[-L,L]\subset \Z$, $S\ne \emptyset$,  we define projections $P_{\pm}{\up{S}}$ by
\beq
P_+^{\pa{S}}= \bigotimes_{j\in S}\ \tfrac{1}{2} (1+\sigma_j^z)\qtx{and} P_-^{\pa{S}}=1-P_+^{\pa{S}}.
\eeq 
Note that
\beq\label{PsuppX}
 P_-^{(S)}\le \sum_{i\in S}  \cN_i .
 \eeq 
In particular,  
\beq\label{P-SP}
P_-^{\pa{S}}P_0 = P_0 P_-^{\pa{S}}=0 .
\eeq
We also set $S^c= [-L,L]\setminus S$, and note that
\beq\label{PPc}
P_+^{\pa{S}}P_+^{\pa{S^c}}=P_+^{\pa{S^c}}P_+^{\pa{S}}= P_+^{[-L,L]}= P_0.
\eeq

Given an observable $X$, we set $P_\pm ^{\pa{X}}= P_{\pm}^{(\cS_X)}$, obtaining the decomposition
\beq\label{Xdecomp}
 X =\sum_{a,b \in \set{+,-}}X^{a,b}, \qtx{where} X^{a,b}= P_{a} ^{\pa{X}} X P_{b} ^{\pa{X}}.
 \eeq
Moreover,  since $P_+^{\pa{X}}$ is a rank one projection on $\cH_{\cS_X}$, we must  have 
\beq\label{Xzeta}
X^{+,+}=\zeta_X P_+^{\pa{X}}, \qtx{where} \zeta_X\in \C, \ \abs{\zeta_X} \le \|X\|.
\eeq
 In particular, 
 \beq \label{X++0}
 \pa{X- \zeta_X}^{+,+}=0 \qtx{and} \norm{X- \zeta_X}\le 2 \norm{X}.
 \eeq

\subsection{Consequences of Property DL}
 Let  $H_\omega$  be a disordered  XXZ spin chain satisfying Property DL.
We write  $H=H_\omega^{(L)}$, and  generally omit   $\omega$ and $L$  from the notation. The following results hold uniformly on $L$.

\begin{lemma}\label{lem:cordyn} Let $X,Y$ be local observables.   Then
\begin{align}\label{P-gP-}
&\E \pa{\sup_{g\in G_{I_0}} \norm{P_{-} ^{\pa{X}} g(H)P_{-} ^{\pa{Y}}}_1}\le C \e^{-m\dist (X,Y)},
\\ \label{P-gP-2}
& \E\pa{\norm{P_{-}^{\pa{Y}} P_{-} ^{\pa{X}} P_{I_0} }_1 } \le C \e^{-\frac 1 2 m \dist (X,Y)}.
\end{align}
\end{lemma}

\begin{proof}  It  follows from \eq{PsuppX} that setting
$Z=   \pa{ \sum_{i\in S}  \cN_i }^{-1}P_{-} ^{\pa{S}}$, we have $\norm{Z}\le 1$ and 
$P_{-} ^{\pa{S}}=\pa{ \sum_{i\in S}  \cN_i } Z= Z \pa{ \sum_{i\in S}  \cN_i }$, and hence we have
\beq
\norm{P_{-} ^{\pa{X}} g(H)P_{-} ^{\pa{Y}}}_1\le \sum_{i\in \cS_X,\,  j\in \cS_Y} \norm{\cN_i g(H)\cN_j}_1.
\eeq
The estimate \eq{P-gP-} then follows immediately from from \eq{eq:efcorDL}  using \cite[Eq.~(3.25)]{EKS}

Similarly,
\begin{align}
\norm{P_{-} ^{\pa{Y}} P_{-} ^{\pa{X}} P_{I_0}}_1 = \norm{P_{-} ^{\pa{Y}} P_{-} ^{\pa{X}} P_{I}}_1\le \sum_{k=-L}^L \norm{P_{-} ^{\pa{Y}} P_{-} ^{\pa{X}} P_I\cN_k }_1.
\end{align}
Since $[P_{-} ^{\pa{Y}} ,P_{-} ^{\pa{X}}]=0$,
\begin{align}
\norm{P_{-} ^{\pa{Y}} P_{-} ^{\pa{X}} P_I\cN_k }_1\le \min \set{\norm{ P_{-} ^{\pa{X}} P_I\cN_k }_1,\norm{P_{-} ^{\pa{Y}} P_I\cN_k }_1},
\end{align}
so it follows from \eq{P-gP-} that
\beq
\E\pa{\norm{P_{-} ^{\pa{Y}} P_{-} ^{\pa{X}} P_I\cN_k }_1}\le C \e^{-m \max \set{ \dist\pa{k, \cS_X}, \dist\pa{k, \cS_Y}}}.
\eeq
Suppose, say, $\max \cS_X <\min \cS_Y $, and    let $K= \frac 12 \pa{\max \cS_X +  \min \cS_Y}$. Then, 
\begin{align}\nn
\E\pa{\norm{P_{-} ^{\pa{Y}} P_{-} ^{\pa{X}} P_I}_1}&\le \sum_{k\le K}  \e^{-m\dist\pa{k, \cS_Y}} +
\sum_{k\ge  K}  \e^{-m\dist\pa{k, \cS_X}}\\
&  \le C \e^{-\frac 1 2 m \dist (X,Y)},
\end{align}
where the last calculation is done as in \cite[Eq. (3.25)]{EKS}, yielding \eq{P-gP-2}. 
\end{proof}

\begin{lemma}\label{lem:cordyn444}Let $X$ and $Y$ be  local observables and $\ell \ge 1$.

 \begin{enumerate}
 
 \item   We  have 
\begin{align}\label{P-gP-3}
\E \pa{\sup_{I\in G_I}\norm{P_{-} ^{\pa{X}}g(H) P_{+} ^{\pa{\cS_{X,\ell}}}}_1} \le C \e^{-m \ell}.
\end{align}

\item  If   $ \ell  \le \frac 1 2\dist (X,Y)$,
we have
\begin{align}\label{P-gP-4}
\E \pa{\sup_{g\in G_I}\norm{P_{+} ^{\pa{\cS_{Y,\ell}^c}}g(H)P_{+} ^{\pa{\cS_{X,\ell}^c}}}_1}\le  C \e^{-m \pa{\dist \pa{X,Y}-2\ell}}.\end{align}

\end{enumerate}

\end{lemma}

\begin{proof}
 Let  $\ell \ge 1$ and   $g \in G_I$.  If $\cS_{X,\ell}^c= \emptyset$, \eq{P-gP-3} is obvious since $P_{+} ^{\pa{\cS_{X,\ell}}}=P_0$.  If $\cS_{X,\ell}^c\not= \emptyset$, using \eq{PPc} we get
  \begin{align}\label{PPc1}
&\norm{P_{-} ^{\pa{X}}g(H) P_{+} ^{\cS_{X,\ell}}}_1 =\norm{P_{-} ^{\pa{X}}g(H) P_{-} ^{\cS_{X,\ell}^c} P_{+} ^{\cS_{X,\ell}}}_1 \le
\norm{P_{-} ^{\pa{X}}g(H) P_{-} ^{\cS_{X,\ell}^c} }_1,
\end{align}
and   \eq{P-gP-3} follows from \eq{PPc1} and  \eq{P-gP-}.

Similarly, using \eq{PPc} twice, we get
\begin{align}\notag
\norm{P_{+} ^{\pa{\cS_{Y,\ell}^c}}g(H)P_{+} ^{\pa{\cS_{X,\ell}^c}}}_1 &  = \norm{P_{+} ^{\pa{\cS_{Y,\ell}^c}}P_{-} ^{\pa{\cS_{Y,\ell}}}g(H) 
P_{-} ^{\pa{\cS_{X,\ell}}}P_{+} ^{\pa{\cS_{X,\ell}^c}}}_1\\
& \le \norm{P_{-} ^{\pa{\cS_{Y,\ell}}}g(H) 
P_{-} ^{\pa{\cS_{X,\ell}}}}_1. \label{P-gP-4222}
\end{align}
If  $ \ell   \le\frac 1 2 \dist (X,Y)$, then $\dist (\cS_{X,\ell},\cS_{Y,\ell})\ge  \dist (X,Y) - 2\ell $.
  In this case \eq{P-gP-4} follows from  \eq{P-gP-4222} and \eq{P-gP-}.
\end{proof}

\begin{lemma} \label{lemPXtgY}  Let $X,Y$ be local observables with $X^{+,+}=Y^{+,+}=0$.
Then
\beq\label{PXtgY}
\E\pa{\sup_{t\in \R}\sup_{g\in G_I}\norm{\pa{\tau_t\pa{ X}g(H)Y}_I}_1}\le  C\norm{X}\norm{Y}  \e^{-\frac 1 8 m\dist (X,Y)} .
\eeq
\end{lemma}

\begin{proof} Since
 \beq
 \norm{\pa{\tau_t\pa{ X}g(H)Y}_I}_1= \norm{\pa{{ X}\e^{-itH}g(H)Y}_I}_1 ,
 \eeq
it suffices to prove
\beq\label{PXtgY99}
\E\pa{\sup_{g\in G_I}\norm{\pa{ X g(H)Y}_I}_1}\le  C\norm{X}\norm{Y}  \e^{-\frac 1 8 m\dist (X,Y)} .
\eeq

 Let $X,Y$ be local observables with $X^{+,+}=Y^{+,+}=0$, and let  $0<2 \ell=    \dist (X,Y) $. Set $\cS_1=\cS_{X,\frac \ell 2}^c$, $\cS_2=\cS_{Y,\frac \ell 2}^c$.   Given  $g \in G_I$, and  inserting  $1=P_{-} ^{\pa{\cS_j}}+P_{+} ^{\pa{\cS_j}}$, $j=1,2$, we get
\beq
 X g(H)Y=\sum_{a=\pm;b=\pm} { XP_{a} ^{\pa{\cS_1}}}g(H)P_{b} ^{\pa{\cS_2}}Y.
\eeq
We estimate the norms of the terms on the right hand side separately. If one of the indices $a,b$, say $a=-$, we get  
\begin{align}\nn
&\norm{\pa{XP_{-} ^{\pa{\cS_1}}g(H)P_{b} ^{\pa{\cS_2}}Y}_I}_1 \le  \|Y\|\norm{P_I XP_{-} ^{\pa{\cS_1}}\e^{-itH}g(H)}_1 \\ & \nn \quad \le  \|Y\|\norm{P_I XP_{-} ^{\pa{\cS_1}}P_I}_1
= \|Y\|\norm{P_I P_{-} ^{\pa{\cS_1}} XP_{-} ^{\pa{\cS_1}}P_I}_1\\ &  \quad \le \norm{X} \|Y\| \pa{\norm{P_IP_{-} ^{\pa{\cS_1}}P_{-} ^{\pa{X}}}_1 +\norm{P_{-} ^{\pa{X}} P_{-} ^{\pa{\cS_1}}P_I}_1},\label{EPX-+-Y49}
\end{align}
where we have used the fact that $[P_{-} ^{\pa{\cS_1}},X]=0$,  $X^{+,+}=0$, and $g\in G_I$. If  $a=b=+$,  we bound the corresponding contribution as 
\beq
\norm{ \pa{ XP_{+} ^{\pa{\cS_1}} g(H)P_{+} ^{\pa{\cS_2}}Y}_I}_1\le \|X\|\|Y\|\norm{P_{+} ^{\pa{\cS_1}}g(H)P_{+}^{\pa{\cS_2}}}_1.
\eeq
Using  \eqref{P-gP-2} and \eqref{P-gP-4} we get 
\begin{align}\nn
E\pa{\sup_{g\in G_I} \norm{\pa{{ X}g(H)Y}_I}_1} &\le C\norm{X}\norm{Y}\pa{2 \e^{- \frac m 4 \ell} + \e^{- m \pa{\dist \pa{X,Y}-\ell }}}\\ & \le   C\norm{X}\norm{Y}\e^{-\frac 1 8 m\dist \pa{X,Y}} . \label{EPX-+-Y}
\end{align}
 \end{proof}

 The following lemma justifies Remark~\ref{remwonder}.
 
 \begin{lemma}\label{lemwonder}  Let $X,Y$ be local observables. Then for all  intervals $K \subset I$
 we have
 \begin{align}\label{decaycounter}
 E \pa{\sup_{t\in \R}\abs{\tr\pa{\tau^K_t(X)P_0 Y}_K }}\le C \|X\| \|Y\| \e^{-m\dist (X,Y)}.
 \end{align}
 \end{lemma}
 
 \begin{proof}
 Given $K\subset I$, we have
 \begin{align}
\tr\pa{\tau^K_t(X)P_0 Y}_K & = \tr P_K\tau_t(X)P_0 YP_K = \tr P_0 YP_K  \tau_t(X)P_0\\ \nn & = \tr P_0 Y P_{-} ^{\pa{Y}}P_K  \e^{itH} P_{-} ^{\pa{X}}XP_0,
\end{align}
where we used \eq{Xzeta}, \eq{P-SP}, and $P_K P_0=0$.  It follows that
\begin{align}
\abs{\tr\pa{\tau^K_t(X)P_0 Y}_K}\le \|X\| \|Y\|  \norm{ P_{-} ^{\pa{Y}}P_K  \e^{itH} P_{-} ^{\pa{X}}}_1.
\end{align}
The estimate \eq{decaycounter} now follows from  \eq{P-gP-}.
 \end{proof}

\subsection{Estimates with Fourier transforms}
  
Let  $H_\omega$  be a disordered  XXZ spin chain.
Given a function  $f\in C^\infty_c(\R)$, we write its Fourier transform as
 \beq\label{FTf}
\hat f(t)=\tfrac 1 {2\pi} \int_\R \e^{itx}  f(x) \, \d x,\qtx{and recall } f(x)= \int_\R \e^{-itx} \hat f(t) \, \d t .
\eeq

The following lemma is an adaptation of an argument of Hastings \cite{Hast,HK}, which combines the Lieb-Robinson bound  with estimates on  Fourier transforms.

\begin{lemma} \label{lemHast} Let $\alpha \in (0,1)$, and consider a function $f\in C^\infty_c(\R)$  such that
\beq
\abs{\hat f(t)} \le C_f \e^{-m_f\abs{t}^\alpha} \qtx{for all} \abs{t}\ge 1,
\eeq
where $C_f$ and $m_f >0$ are constants.
Then for all local observables   $X$ and $Y$ we have
\begin{align}\label{Htrickest}
&\norm{Xf(H)Y -  \int_\R   \e^{-irH} Y \tau_r\pa{ X} \hat f(r) \, \d r }\\ \nn & \qquad  \qquad  \qquad
\le  C_1  \norm{X} \norm{Y}\pa{1+  \norm{\hat f}_1} \e^{- m_1\pa{ \dist (X,Y)}^\alpha} ,
\end{align}
where $C_1$ and $m_1>0$ are suitable constants (depending on $C_f$, $m_f$, and $\alpha$),
uniformly in $L$.
\end{lemma}

\begin{proof}

We have 
\begin{align}
Xf(H)Y&= X \pa{\int_\R \e^{-irH} \hat f(r) \, \d r}Y= \int_\R   \e^{-irH} \tau_r\pa{ X} Y\hat f(r) \, \d  r \\
\nn & =\int_\R   \e^{-irH}[ \tau_r\pa{ X} ,Y] \hat f(r) \, \d r  + \int_\R   \e^{-irH} Y \tau_r\pa{ X} \hat f(r) \, \d r 
\end{align}

The commutator in the  first term can be estimated by the Lieb-Robinson bound (e.g. \cite{NS}):
\beq\label{LRbound}
\norm{[\tau_r\pa{X}, Y]}\le C\norm{X}\norm{Y}  \min \set{\e^{- \mu_1 \pa{\dist (X,Y)-v\abs{r}}},1},
\eeq
where $C$, $\mu_1>0$, $v>0$ are constants, independent of $L$ and of the random parameter $\omega$.  We get 
\begin{align}
&\norm{\int_\R   \e^{-irH}[ \tau_r\pa{ X} ,Y] \hat f(r) \, \d r} \\ \nn
&\le C  \norm{X} \norm{Y}\hskip-2pt\pa{\int_{\abs{r} \le  \frac {\dist (X,Y)}{2v}}\hskip-5pt \e^{- \mu_1 \pa{\dist (X,Y)-v\abs{r}}} \abs{\hat f(r)} \, \d r + \int_{\abs{r} \ge \frac {\dist (X,Y)}{2v}}  \abs{\hat f(r)} \, \d r  \hskip-2pt}\\ \nn
&\le C     \norm{X} \norm{Y} \pa{  \norm{\hat f}_1 \e^{- \frac{\mu_1}2 \dist (X,Y)}  + \int_{\abs{r} \ge \frac {\dist (X,Y)}{2v}}  \abs{\hat f(r)} \, \d r }\\ \nn
&\le C     \norm{X} \norm{Y} \pa{  \norm{\hat f}_1 \e^{- \frac{\mu_1}2 \dist (X,Y)}  
+ C_f  \e^{-\frac {m_f}{2}\pa{\frac {\dist (X,Y)}{2v}}^\alpha}
\int_{\R}  \e^{-\frac {m_f}2 \abs{r}^\alpha} \, \d r   },
\end{align}
where we assumed 
${\dist (X,Y)}\ge{2v}$.   The estimate \eq{Htrickest} follows.
\end{proof}

Lemma~\ref{lemHast} will be combined with the following lemma. 

\begin{lemma}\label{leminsertKf} Let $K=[\Theta_0,  \Theta_2]$ and   $f\in C^\infty_c(\R)$ with $\supp f \subset [a_f,b_f]$. Then for all local observables   $X$ and $Y$ we have
\begin{align}\label{KKfK}
\int_\R   \pa{  \e^{-irH} Y \tau_r\pa{ X} }_K \hat f(r) \, \d r  = \int_\R  \pa{  \e^{-irH} Y P_{K_f} \tau_r\pa{ X}}_K \hat f(r) \, \d r,
\end{align}
where 
\beq\label{Kf}
K_f =  K + K -\supp f \subset [2\Theta_0 - b_f, 2\Theta_2- a_f     ].
\eeq
\end{lemma}

\begin{proof}
Let $K=[\Theta_0,  \Theta_2]$,  $f\in C_c(\R)$ with $\supp f \subset [a_f,b_f]$. Then for all $E,E^\pr \in K$ we have
\begin{align}\nn
& P_E \pa{ \int_\R   \e^{-irH} Y \tau_r\pa{ X} \hat f(r) \, \d r }  P_{E^\pr}=\int_\R  P_E\,  \e^{-irH} Y \e^{irH}{ X} \e^{-irH}P_{E^\pr}\hat f(r) \, \d r\\  & \quad = P_E Y \pa{ \int_\R  \e^{ir(H-E -E^\pr)}\hat f(r) \, \d r}{ X} P_{E^\pr}=  P_E Y f(E+E^\pr - H)   { X} P_{E^\pr} \nn  \\ & \quad = \nn
 P_E Y P_{K_f} f(E+E^\pr - H)   { X} P_{E^\pr} \\ & \quad =  P_E \pa{ \int_\R   \e^{-irH} YP_{K_f} \tau_r\pa{ X} \hat f(r) \, \d r }  P_{E^\pr},\label{P0mag}
\end{align}
where $K_f$ is given in \eq{Kf}. The equality \eq{KKfK} follows.
\end{proof}

\subsection{Counterterms} Given vectors $\psi_1,\psi_2\in \cH\up{L}$, we denote by $T(\psi_1,\psi_2)$  the rank one operator  $T(\psi_1,\psi_2)= \scal{\psi_2, \cdot}\psi_1$.  Recall  
\[ \norm{T(\psi_1,\psi_2)} =\norm{T(\psi_1,\psi_2)}_1=\norm{\psi_1}\norm{\psi_2}.\]
Note that for all observables
 $X$ and $Y$ we have 
\begin{align}
XP_0\up{L} Y& =  T\pa{Y^*\psi_0\up{L}, X\psi_0\up{L}}. \label{Tnotation}
\end{align}

\begin{lemma}\label{lem:spillterms} Let  $H_\omega$  be a random  XXZ spin chain. Consider an interval  $K\subset   [1 -\tfrac 1 \Delta, 1 +\tfrac 1 \Delta]$.  Then there exist constants $\gamma_K >0$  and $R_K$ such that for all  $i,j \in \Z$ with $\abs{i-j}\ge R_K$, we have
\begin{align}\label{eq:1term}
 \E\pa{\liminf_{L\to \infty}\norm{ \pa{\sigma^x_i P\up{L}_0\sigma^x_j}_K}}
 \ge \gamma_K>0,
\end{align}
\beq\label{eq:2terms0}
\E\pa{\liminf_{L\to \infty}\norm{ \pa{\ \sigma^x_i P\up{L}_0\sigma^x_j \pm  \sigma^x_j P\up{L}_0\sigma^x_i}_K}_{2}^2}\ge \gamma_K,
\eeq
and
\begin{align}\label{eq:4terms1}
\E\pa{\liminf_{L\to \infty} \lim_{T\to \infty} \tfrac 1 T \int_0^T \norm{\pa{A\up{L}(t)  -\pa{A\up{L}(t)}^*}_K}_2^2\, \d t} \ge 2 \gamma_K,
\end{align}
where
\begin{align}\label{eq:4terms12}
 A\up{L}(t)& =  \tau\up{L}_t\!\pa{ \sigma^x_i  }P_0\up{L}\sigma^x_j  +\tau\up{L}_t\!\pa{\sigma^x_j } P_0\up{L} \sigma^x_i  . \end{align}
\end{lemma}

\begin{proof}Let $H$ be a random XXZ spin chain, and  let $\mathcal{N} = \sum_{i\in \Z} \mathcal{N}_i$ denote the total (down) spin number operator on $\cH$. The self-adjoint operator $\mathcal{N}$ has pure point spectrum.  Its eigenvalues are $N=0,1, 2, \ldots$, and the corresponding eigenspaces $\cH_N$  are spanned by all the spin basis states with $N$ down spins.  Since $[H, \mathcal{N}]=0$, the eigenspaces $\cH_N$ are left invariant by $H$. The  restriction $H_N$ of $H$ to 
 $\cH_N$ is unitarily equivalent to an $N$-body discrete Schr\"odinger operator restricted to the fermionic subspace (e.g., \cite{FS,EKS}).

In particular,  $H_1=H_{\omega,1}$ is unitarily equivalent to an one-dimensional  Anderson model:
\begin{align}\label{AndH}
H_{\omega,1} \cong  -\tfrac{1}{2\Delta}\mathcal{L}_1+\pa{1-\tfrac{1}{\Delta}} +\lambda V_\omega
\qtx{on}\ell^2(\Z),
\end{align}
where $\cL_1$ is the graph Laplacian on  $\ell^2(\Z)$ and $V_\omega$ is the random potential given by $V_\omega(i)=\omega_i$ for $i \in \Z$. 

The same is true for restrictions to finite intervals $[-L,L]$, where we have the unitary equivalence  
\begin{align}\label{AndHL}
H_{\omega,1}\up{L} \cong  -\tfrac{1}{2\Delta}\mathcal{L}_1\up{L}+\pa{1-\tfrac{1}{\Delta}} +\lambda V_\omega +  \left(\beta-\tfrac{1}{2}(1-\tfrac{1}{\Delta})\right)\pa{ \chi_{\set{-L}}+\chi_{\set{L}}},
\end{align}
acting on $\ell^2([-L,L])$,
where now $\cL_1\up{L}$ is the graph Laplacian on  $\ell^2([-L,L])$
 (e.g., \cite{EKS}). Note that $H_{\omega,1}\up{L}$ is the restriction of $H_{\omega,1}$ to 
 $\ell^2([-L,L])$, up to a boundary term.
 
In what follows we will consider these unitary equivalences as equalities. In this case,
if  $i\in [-L,L]$ we have   $\sigma^x_i \psi_0\up{L} =\delta_i \in \ell^2([-L,L])$,   Note that  for the infinite volume Anderson model in \eq{AndH} we have
\beq
\sigma\pa{H_1}\supset  \Sigma_1:= [1 -\tfrac 1 \Delta, 1 +\tfrac 1 \Delta] \quad \text{almost surely}.
\eeq 

 The following holds for all
$\omega \in [0,1]^{\Z}$:   We have  $\lim_{L\to \infty} H_1\up{L}= H_1$ in the strong resolvent sense, and hence  $\lim_{L\to \infty} f\pa{H_1\up{L}}= f\pa{H_1}$ strongly for all bounded continuous functions $f$ on $\R$. (For an interval $J\subset \Z$, we consider $\ell^2(J)$ as a subspace of $ \ell^2(\Z)$ in the obvious way: $ \ell^2(\Z) =\ell^2(J)\oplus \ell^2(\Z\setminus J)$.)
  In particular, for $f$ real valued with $\norm{f}_\infty\le 1$, 
\beq\label{everywherelim}
\sup_L \norm{ f(H_1\up{L}) \delta_u}\le 1 \sqtx{and}  \lim_{L\to \infty} f\pa{H_1\up{L}}\delta_u = f\pa{H_1}\delta_u \sqtx{for all} u \in \Z.
\eeq
 Moreover, 
\begin{align}\nn
\lim_{L\to \infty}  \E\pa{ \norm{ f(H_1\up{L}) \delta_u}^2}& =\E\pa{ \norm{ f(H_1) \delta_u}^2}
=\E\pa{ \scal{\delta_u, \pa{f(H_1)}^2 \delta_u}}\\
&  = \E\pa{ \scal{\delta_0, \pa{f(H_1)}^2 \delta_0}}= \int f^2(t) \, \d\eta(t), \label{conveta}
\end{align}
where $\eta$ is the density of states measure for the Anderson model $H_1$.
It also   follows from \eq{everywherelim} by bounded convergence that
\begin{align}\label{limLinfty}
\lim_{L\to \infty} \E\pa{\norm{ f(H_1\up{L}) \delta_j} \norm{  f(H_1\up{L}) \delta_i}}=\E\pa{\norm{ f(H_1) \delta_j} \norm{  f(H_1) \delta_i}}.
\end{align}

We now fix a function $f\in C_c(\R)$ such that  $\supp f \subset K\cap \Sigma_1$ and $\chi_{K^{\pr}} \le f \le\chi_{K\cap \Sigma_1}$ for some nonempty interval $K^\pr \subset K\cap \Sigma_1$. Note that
\beq\label{convetaD}
D:= \int f^2(t) \, \d\eta(t) \ >0,
\eeq

   Given $i,j\in \Z$,  if  $i,j\in [-L,L] $, we have  
\begin{align}\nn
&\norm{ \pa{\sigma^x_i P\up{L}_0\sigma^x_j}_K}=\norm{ P_K\up{L}  \sigma^x_j \psi_0\up{L}} \norm{ P_K\up{L} \sigma^x_i \psi_0\up{L}}\\ &\qquad \quad  =\norm{ P_K\up{L} \delta_j} \norm{ P_K\up{L}  \delta_i}\ge 
\norm{ f(H_1\up{L}) \delta_j} \norm{  f(H_1\up{L}) \delta_i}, \label{intf}
\end{align}
and  hence it follows from \eq{everywherelim} that
\begin{align}
&\liminf_{L\to\infty}\norm{ \pa{\sigma^x_i P\up{L}_0\sigma^x_j}_K} \ge 
\norm{ f(H_1) \delta_j} \norm{  f(H_1) \delta_i}, \label{intf58}
\end{align}

Given $u\in \Z$, let $H_1\up{u,L}$ denote the restriction of $H_1$ to the interval $[u-L,u+L]=u +[-L,L]$, and note that \eq{everywherelim} and \eq{conveta} hold with $H_1\up{u,L}$ substituted for $H_1\up{L}=H_1\up{0,L}$.  In particular,
 \begin{align}\label{epsL0}
\lim_{L\to \infty} \eps\up{u,L}=0, \qtx{where} \eps\up{u,L}= \E\pa{\norm{ \pa{f(H_1\up{u,L})-f(H_1)} \delta_u}},
 \end{align}
and note that $\eps_L= \eps\up{u,L}$ is independent of $u\in \Z$. Moreover,
\beq\label{1>2}
\E\pa{\norm{ f(H_1 )\delta_u} }\ge \E\pa{\norm{ f(H_1 )\delta_u}^2 }.
\eeq

It follows that for all  $i,j\in \Z$ and $ L\in \N$, with $ \abs{i-j} \ge 3L$ we have ($\eps_L \le 1$)
\begin{align}\nn 
&\E\pa{\norm{ f(H_1) \delta_j} \norm{  f(H_1) \delta_i}}
\ge \E\pa{\norm{ f(H_1\up{j,L}) \delta_j} \norm{  f(H_1\up{i,L}) \delta_i}}- 2 \eps_L \\ \nn &\quad \quad
= \E\pa{\norm{ f(H_1\up{j,L}) \delta_j} } \E\pa{ \norm{  f(H_1\up{i,L}) \delta_i}}- 2\eps_L 
\\ \nn &\quad \quad
\ge \E\pa{\norm{ f(H_1 )\delta_j} } \E\pa{ \norm{  f(H_1) \delta_i}}- 4\eps_L 
\\ \nn &\quad \quad
\ge \E\pa{\norm{ f(H_1 )\delta_j}^2 } \E\pa{ \norm{  f(H_1) \delta_i}^2}-4\eps_L \\  &\quad \quad
= \E\pa{\norm{ f(H_1 )\delta_0}^2 }^2- 4\eps_L \ge D^2 - 4\eps_L\ge \tfrac 12 D^2 \label{Effij}
 \end{align}
 where we used \eq{epsL0},  the fact that the collections of random variables $\set{\omega_k}_{k\in [j-L,j+L]}$ and $\set{\omega_s}_{s\in [i-L,i+L]}$ are   independent, used \eq{epsL0} again,  used \eq{1>2}, and   the last inequality follows from \eq{conveta}, \eq{convetaD}, and \eq{epsL0}, taking $L$ sufficiently large.  In particular, there exists $\tilde R$ such that \eq{Effij}
holds if $\abs{i-j} \ge \tilde R$.

It follows from \eq{intf58}  and \eq{Effij} that for $\abs{i-j} \ge \tilde R$ we have
\begin{align}
 \E \pa{ \liminf_{L\to \infty}\pa{\norm{ P\up{L}_K  \sigma^x_j \psi_0\up{L}} \norm{ P_K\up{L}  \sigma^x_i \psi_0\up{L}}}} \ge \tfrac 12 D^2,
\end{align}
which is  \eq{eq:1term}.

Note that   $\sqrt{f} \in C_c(\R)$  and  $\chi_{K^{\pr}} \le f\le \sqrt{f}  \le\chi_{K\cap \Sigma_1}$.  Given an observable $X$ we have
\begin{align}\nn
&\norm{X_K}_{2}^2=\norm{P_K\up{L}XP_K\up{L}}_{2}^2=\tr \pa{P_K\up{L}X^*P_K\up{L}X P_K\up{L}} \\ &\quad\nn \ge \tr \pa{P_K\up{L}X^*f(H\up{L})X P_K\up{L}}=  \tr \pa{\sqrt{f}(H\up{L})X P_K\up{L}X^*\sqrt{f}(H\up{L})}\\ &\quad\ge   \tr \pa{\sqrt{f}(H\up{L})X f(H\up{L})X^*\sqrt{f}(H\up{L})}  = \norm{\sqrt{f}(H\up{L})X \sqrt{f}(H\up{L})}_{2}^2. \label{eq:Ptof}
\end{align}
Thus, we can estimate 
\begin{align}\nn
&\norm{ \pa{\sigma^x_i P_0\up{L}\sigma^x_j \pm  \sigma^x_j P_0\up{L}\sigma^x_i}_K}_{2}^2  \\ & \nn \quad \ge \norm{\sqrt{f}({H_1\up{L}}) \pa{\sigma^x_i P_0\up{L}\sigma^x_j \pm   \sigma^x_j P_0\up{L}\sigma^x_i}\sqrt{f}({H_1\up{L}})}_{2}^2\\ & \nn \quad  = \norm{T\pa{\sqrt{f}({H_1\up{L}}) \delta_i, \sqrt{f}({H_1\up{L}})\delta_j}\pm  T\pa{ \sqrt{f}({H_1\up{L}}) \delta_j,\sqrt{f}({H_1\up{L}})\delta_i}}_2^2  \\ & \quad  = \nn
2 \pa{\norm{ \sqrt{f}({H_1\up{L}}) \delta_i}^2\norm{ \sqrt{f}({H_1\up{L}}) \delta_j}^2\pm  \Rea  \pa{\scal{\delta_j,  f ({H_1\up{L}})\delta_i }}^2 }\\ & \quad  \ge 2 \pa{\norm{ {f}({H_1\up{L}}) \delta_i}^2\norm{ {f}({H_1\up{L}}) \delta_j}^2- \abs{\scal{\delta_j,  f ({H_1\up{L}})\delta_i }} }.
\label{sps}
\end{align}
It follows from \eq{sps} and  \eq{everywherelim} that
\begin{align}\nn
&\liminf_{L\to \infty} \norm{ \pa{\sigma^x_i P_0\up{L}\sigma^x_j \pm  \sigma^x_j P_0\up{L}\sigma^x_i}_K}_{2}^2  \\ & \qquad \qquad \ge  2 \pa{\norm{ {f}({H_1}) \delta_i}^2\norm{ {f}({H_1}) \delta_j}^2-   \abs{\scal{\delta_j,  f ({H_1})\delta_i }} }.\label{sps2}
\end{align}

Given a scale  $\ell$ and $ \abs{i-j} \ge 3\ell$, we have 
\begin{align}\label{sps45}
\abs{\scal{\delta_j,  f (H_1)\delta_i }}= \abs{\scal{\delta_j, \pa{ f (H_1)-f(H_1\up{i,\ell})}\delta_i }}\le
\norm{\pa{ f (H_1)-f(H_1\up{i,\ell})}\delta_i }
\end{align}
Since $\E\pa{\norm{ {f}(H_1) \delta_i}\norm{ {f}(H_1) \delta_j}}\le \pa{\E\pa{\norm{ {f}(H_1) \delta_i}^2\norm{ {f}(H_1) \delta_j}^2}}^{\frac 12}$,
it follows from \eq{sps2},  \eq{sps45},  \eq{Effij} and  \eq{epsL0}, that there exists $\ell_1$, such that for   $ \abs{i-j} \ge 3\ell_1$  we have
\begin{align}\nn
&\E \pa{\liminf_{L\to \infty} \norm{ \pa{\sigma^x_i P_0\up{L}\sigma^x_j \pm  \sigma^x_j P_0\up{L}\sigma^x_i}_K}_{2}^2} \\ & \nn \quad \ge 2  \pa{\pa{\E\pa{\norm{ {f}(H_1) \delta_i}\norm{ {f}(H_1) \delta_j}}}^2 -\E\pa{\norm{\pa{ f (H_1)-f(H_1\up{i,\ell_1})}\delta_i }} } \\ &  \quad \ge 2 \pa{ \tfrac 1 4 D^4 -\eps_{\ell_1}}\ge  \tfrac 1 4 D^4 .
\end{align}
The estimate \eq{eq:2terms0} is proven.

Now let  $A\up{L}(t)$ be as in \eq{Zt} (we mostly omit $L$ from the notation), and let
\begin{align}\nn
Z\up{L}(t) &=   \pa{A\up{L}(t)  -\pa{A\up{L}(t)}^*}_K  \\
&= \e^{itH} \pa{ \sigma^x_i  P_0\sigma^x_j  +\sigma^x_j  P_0 \sigma^x_i }_K - \pa{ \sigma^x_i  P_0\sigma^x_j  +\sigma^x_j  P_0 \sigma^x_i }_K\e^{-itH} \nn  \\
&= \e^{itH} A_K - A_K\e^{-itH} = B_t - B^*_t, \label{Zt}
\end{align}
where
\begin{align}
A= A\up{L}(0)= \sigma^x_i  P_0\sigma^x_j  +\sigma^x_j  P_0 \sigma^x_i = A^*\qtx{and} B_t= \e^{itH} A_K.
\end{align}
We have
\begin{align}
\norm{Z\up{L}(t)}_2^2& = \norm{B_t - B^*_t}_2^2\nn \\
&   =\tr \pa{B_tB^*_t}+\tr \pa{B^*_tB_t} - \tr \pa{B_tB_t} - \tr \pa{B^*_t B^*_t}\nn \\
&   =2 \norm{A_K}^2_2 - 2 \Rea \tr \pa{ P_K\e^{itH} A P_K\e^{itH} AP_K}.
\end{align}
Since
\begin{align}
\tr \pa{ P_K\e^{itH} A P_K\e^{itH} AP_K}= \sum_{E,E^\pr \in \sigma_K}  \e^{it(E+ E^\pr)} \tr \pa{ P_EA P_{E^\pr} AP_E},
\end{align}
 and $0\notin K$, and $\lim_{T\to \infty} \tfrac 1 T \int_0^T \e^{it s} \, \d t =0$ if $s\ne 0$,we conclude that
 \begin{align}
\lim_{T\to \infty} \tfrac 1 T \int_0^T \norm{Z\up{L}(t)}_2^2\, \d t = 2 \norm{A_K}^2_2= 2 \norm{\pa{\sigma^x_i  P_0\sigma^x_j  +\sigma^x_j  P_0 \sigma^x_i}_K}^2_2.
 \end{align}
The estimate \eq{eq:4terms1} now follows from \eq{eq:2terms0}
\end{proof}

\section{Optimality of the droplet spectrum}\label{secopt}
We are ready to prove Theorem~\ref{thmrigid}.

 \begin{proof}[Proof of Theorem~\ref{thmrigid}]
 
 Suppose Property DL  is valid for a disordered XXZ spin chain $H$
with $\Theta_1 > 2 \Theta_0$.  Let    $K=  [\Theta_0, \Theta_2]$, where  $ \Theta_0 < \Theta_2 < \Theta_1 $,  and $\eps= \min \set{ \Theta_1 - 2 \Theta_2, \Theta_0} >0$.   We pick and fix  a Gevrey class function $h$ such that 
  \[0\le h\le 1,\; \supp h \subset (-\eps,\eps),\; h(0)=1,  \sqtx{and} \abs{\hat h(t)}\le C \e^{-c\abs{t}^{\frac 12}}\sqtx{for all} t\in\R,\]
in particular,  $ \norm{\hat h}_1<\infty$.   Note that $P_0= h(H)$.

 Let $X,Y$ be local observables with  $X^{+,+}=Y^{+,+}=0$.   It follows from Lemmas~\ref{lemHast} and \ref{leminsertKf} that 
 \begin{align}
\norm{ \pa{XP_0Y}_K} &= \norm{ \pa{Xh(H)Y}_K}\\ \nn &
\le  C  \norm{X} \norm{Y}\e^{- m_1\pa{ \dist (X,Y)}^{\frac 12}} +C^\pr  \sup_{r\in \R}\norm { \pa{Y P_{K_h} \tau_r\pa{ X} }_K} ,
 \end{align}
 where 
 \beq
 K_h \subset  [2\Theta_0 - \eps, 2\Theta_2+ \eps  ]\subset  [\Theta_0,\Theta_1]=I.
 \eeq
 
It follows from Lemma~\ref{lemPXtgY}  that
 \begin{align}\nn
  \E\pa{\sup_{r\in \R}\norm { \pa{Y P_{K_h} \tau_r\pa{ X} }_K}}&\le \E\pa{\sup_{r\in \R}\norm { \pa{Y P_{K_h} \tau_r\pa{ X} }_I}}\\ & \le C\norm{X}\norm{Y}  \e^{-\frac 1 8 m\dist (X,Y)} ,
 \end{align}
so we conclude that 
\begin{align}\label{estcont}
\E \pa{\norm{ \pa{XP_0Y}_K}}\le  C\norm{X}\norm{Y} \e^{- m_2 \pa{ \dist (X,Y)}^{\frac 12}},
\end{align}
where $m_2= \min\set{m_1,\frac 1 8 m}>0$.

For all $k\in \Z$ we have $\sigma^x_k=\pa{\sigma^x_k}^*$,  $\pa{\sigma^x_k}^{+,+}=0$, and $\norm{\sigma^x_k}=1$.  Thus it follows from \eq{estcont} that for all $i, j \in [-L,L]$ we have (we put $L$ back in the notation)
\begin{align}\label{estcont9}
\E \pa{\norm{ \pa{ \sigma^x_i P\up{L}_0 \sigma^x_j }_K}}\le    C\e^{- m_2 \pa{ \abs{i-j}}^{\frac 12}},
\end{align}
uniformly in $L$.
  
If $H$ is a random XXZ spin chain, \eq{estcont9} contradicts \eq{eq:1term} in  Lemma~\ref{lem:spillterms}
if $\abs{i-j}$ is sufficiently large.  Thus we conclude that we cannot have  $\Theta_1> 2\Theta_0$, that is, we must have $\Theta_1\le2\Theta_0$.
\end{proof}

\section{Non-spreading of information}\label{secnonsp}

In this section we  prove Theorem~\ref{corquasloc}.

\begin{proof}[Proof of Theorem \ref{corquasloc}]
Let  $H_\omega$  be a disordered  XXZ spin chain satisfying Property DL. Let
$X$ be a local observable with support $\cS=\cS_X=[s_X,r_X]$.  In view of  \eq{X++0} we can assume $X^{+,+}=0$.

We take $\ell \ge 1$, and
set  (recall \eq{cSell})
\begin{align}{\mathcal O}&=[-L,L]\setminus {\cS}_{\frac \ell 2} = [-L, s_X - \tfrac \ell 2) \cup  (r_X +\tfrac \ell 2,L] \\ \nn
 \mathcal T& = {\cS}_{ \ell }\cap {\mathcal O}= [s_X-\ell, s_X - \tfrac \ell 2)\cup  (r_X +\tfrac \ell 2,r_X+\ell]
 \end{align}

 We start by proving that
\beq\label{compproof2999}
\E \pa{\sup_{t\in \R}\norm{ \pa{{P_+^{\pa{\mathcal O}}\tau_t\pa{ X_{I_0} }P_+^{\pa{\mathcal O}}} - \tau_t\pa{ X }}_{I_0}}_1} \le C  \|X\|\e^{- \frac{1}{16} m\ell}.
 \eeq
 Given an observable $Z$,
we write   $Z_{I_0}=Z_1+Z_2+Z_3+Z_4$, where 
\beq
Z_1=P_0 ZP_0;\ Z_2=P_I Z P_I= Z_I; \ Z_3=P_0 ZP_I; \ Z_4=P_I ZP_0.
\eeq
 Since $\pa{X_i}_{I_0}=\pa{X_{I_0}}_i=X_i$ and $\tau_t(X_i)=\pa{\tau_t(X)}_i$ for $i=1,2,3,4$,
  $X_1 = X^{+,+}_1= 0$, and $\pa{X_4}^*=\pa{X^*}_3$,  
to prove  \eq{compproof2999} it suffices to prove
 \beq\label{compproof2}
\E \pa{\sup_{t\in \R}\norm{ \pa{P_+^{\pa{\mathcal O}}\tau_t\pa{ X_{I_0} }P_+^{\pa{\mathcal O}}}_i - \tau_t\pa{ X_i }}_1} \le C  \|X\|\e^{- \frac{1}{16} m\ell}
\eeq
in the cases  $i=2,3$.   
 
If $i=3$, we have 
\begin{align} \nn
&\norm{\tau_t\pa{ X_3 }-\pa{P_+^{\pa{\mathcal O}}\tau_t\pa{ X_{I_0}}P_+^{\pa{\mathcal O}}}_3}_1 =\norm{\pa{\tau_t\pa{ X_{I_0} }-P_+^{\pa{\mathcal O}}\tau_t\pa{ X_{I_0} }P_+^{\pa{\mathcal O}}}_3}_1\\ \nn &\quad  \qquad =\norm{ \pa{\tau_t\pa{ X_{I_0}}P_-^{{\mathcal O}}}_3}_1= \norm{P_0 X P_-^{(X)}\e^{-itH} P_I P_-^{{\mathcal O}} P_I}_1
\\  & \quad \qquad  \le \norm{X}\norm{P_-^{(X)} \e^{-itH}P_I P_-^{{\mathcal O}} }_1, \label{Pplusinsert} 
\end{align} 
where we used $P_0 X=P_0 X P_-^{(X)}$  since  $X^{+,+}=0$.  Thus it  follows from  \eqref{P-gP-} that
\begin{align}
\E\pa{\sup_{t\in \R}\norm{\tau_t\pa{ X_3 }-\pa{P_+^{\pa{\mathcal O}}\tau_t\pa{ X_{I_0}}P_+^{\pa{\mathcal O}}}_3}_1}\le C  \|X\|\e^{- \frac{1}{2} m\ell}.
\end{align}
If  $i=2$, recall that $Z_2=Z_I$. Since $P_I P_+^{\pa{\mathcal O}}P_0 = P_I P_0 =0$, we have 
\beq
\pa{P_+^{\pa{\mathcal O}}\tau_t\pa{ X_{I_0} }P_+^{\pa{\mathcal O}}}_I= \pa{P_+^{\pa{\mathcal O}}\tau_t\pa{ X_I}P_+^{\pa{\mathcal O}}}_I.
\eeq
Thus
\begin{align}\nn 
&\norm{\tau_t\pa{ X_I}-\pa{P_+^{\pa{\mathcal O}}\tau_t\pa{ X_{I_0} }P_+^{\pa{\mathcal O}}}_I}_1\\ & \qquad  \qquad  = \nn  \norm{\pa{\tau_t\pa{ X_I } P_-^{{\mathcal O}}}_I+ \pa{P_-^{{\mathcal O}}\tau_t\pa{ X_I } P_+^{\pa{\mathcal O}}}_I}_1\\ &\qquad  \qquad  \le  \norm{\pa{\tau_t\pa{ X_I } P_-^{{\mathcal O}}}_I}_1+
\norm{\pa{ P_-^{{\mathcal O}}\tau_t\pa{ X_I }}_I}_1  \nn  \\ &\qquad  \qquad =\norm{\pa{\tau_t\pa{ X_I } P_-^{{\mathcal O}}}_I}_1+\norm{\pa{\tau_t\pa{ X^*_I } P_-^{{\mathcal O}}}_I}_1.
\end{align}
Since 
\begin{align}\label{eq:spplfk}
\norm{\pa{\tau_t\pa{ X_I }P_-^{{\mathcal O}}}_I}_1= \norm{\pa{\tau_t\pa{ X } P_I P_-^{{\mathcal O}}}_I}_1,
\end{align}
it follows from  Lemma \ref{lemPXtgY} that
\begin{align}
\E\pa{\sup_{t\in \R}\norm{\tau_t\pa{ X_I}-\pa{P_+^{\pa{\mathcal O}}\tau_t\pa{ X_{I_0} }P_+^{\pa{\mathcal O}}}_I}_1}\le C  \|X\|\e^{- \frac{1}{16} m\ell}.
\end{align}
This finishes the proof of \eqref{compproof2}, and hence of \eq{compproof2999}.

 We now observe that for all observables $Z$ we have
 \beq\label{tildeZ}
P_+^{({\mathcal O})}Z P_+^{\pa{\mathcal O}}=\tilde Z P_+^{\pa{\mathcal O}}=P_+^{\pa{\mathcal O}} \tilde Z ,
\eeq
where $\tilde Z  $ is an observable with     $\cS_{\tilde Z}={\cS}_{\frac \ell 2}$ and $\|\tilde Z  \|\le \|Z\|$. To see this,  we  write the Hilbert space as $\mathcal{H}\up{L} = \mathcal{H}_{\mathcal O} \otimes \mathcal{H}_{\cS_{\frac \ell 2} }$,  and let $\psi_{\mathcal O}= \otimes_{i\in {\mathcal O}} \, e_{+}$ be the all spins up vector in $\mathcal{H}_{{\mathcal O}}$. We define $T: \mathcal{H}_{{\cS}_{\frac \ell 2} }\to \mathcal{H}\up{L}$ by $T \eta= \psi_{\mathcal O} \otimes \eta$ and
$R:\mathcal{H}\up{L} \to \mathcal{H}_{{\cS}_{\frac \ell 2} }$ by $P_+^{(\mathcal{O})} \varphi = \psi_{\mathcal O} \otimes R\varphi$. i.e., $P_+^{(\mathcal{O})} = T R$.  Note $\norm{T},\norm{R}\le 1$.  Given an observable $Z$, we define $\hat Z :\mathcal{H}_{{\cS}_{\frac \ell 2} } \to\mathcal{H}_{{\cS}_{\frac \ell 2} }$  by $ \hat Z= RZT$.  Then   $\tilde Z= I_{\mathcal{H}_{\mathcal O} } \otimes \hat Z$ satisfies \eq{tildeZ}.

It follows from \eq{compproof2999} and \eq{tildeZ} that
\beq\label{compproof2111}
\E \pa{\sup_{t\in \R}\norm{ \pa{P_+^{\pa{\mathcal O}} \widetilde{\tau_t\pa{  X_{I_0}  }} - \tau_t\pa{ X }}_{I_0}}_1} \le C  \|X\|\e^{- \frac{1}{16} m\ell}. 
 \eeq

Since $\widetilde{\tau_t\pa{  X_{I_0}  }}$ does not have support in $\cS_\ell$, 
we now define
\beq\label{eq:X_ell}
X_\ell(t)=P_+^{\pa{\mathcal T}} \widetilde{\tau_t\pa{  X_{I_0}  }} = \widetilde{\tau_t\pa{  X_{I_0}  }} P_+^{\pa{\mathcal T}} \qtx{for} t\in \R,
\eeq
an observable with support   in ${\cS}_{\frac \ell 2} \cup \mathcal T= {\cS}_{\ell }  $.
and claim that $X_\ell(t)$ satisfies \eqref{qlocI}. 

To show that \eqref{qlocI} follows from \eq{compproof2111}, we consider an observable $Y$ with $\cS_Y=\mathcal{O}^c= {\cS}_{\frac \ell 2}$, and note that
\beq \label{replaceOT1}
\pa{P_+^{\pa{\mathcal T}}-P_+^{\pa{\mathcal O}}} Y= P_-^{{\mathcal O}\setminus\mathcal T}P_+^{\pa{\mathcal T}}Y. 
\eeq
Since $P_0P_-^{{\mathcal O}\setminus\mathcal T}=P_-^{{\mathcal O}\setminus\mathcal T}P_0=0$, we have
\begin{align} \label{replaceOT2}
\pa{ P_-^{{\mathcal O}\setminus\mathcal T}P_+^{\pa{\mathcal T}}Y}_{I_0}=\pa{ P_-^{{\mathcal O}\setminus\mathcal T}P_+^{\pa{\mathcal T}}Y}_{I}. 
\end{align}
 
We now apply \eq{replaceOT1} and \eq{replaceOT2} with $Y= \widetilde{\tau_t\pa{  X_{I_0}  }} $.
We have 
\begin{align}\nn
& P_+^{\mathcal{O}} \pa{\widetilde{\tau_t\pa{  X_{I_0}  }}}^{+,+} P_+^{\mathcal{O}}  =P_+^{\mathcal{O}}P_+^{\mathcal{O}^c} \widetilde{\tau_t\pa{  X_{I_0}  }}P_+^{\mathcal{O}^c}P_+^{\mathcal{O}}\\ \nn
& = P_+^{\mathcal{O}^c}P_+^{\mathcal{O}} \widetilde{\tau_t\pa{  X_{I_0}  }}P_+^{\mathcal{O}}P_+^{\mathcal{O}^c} = P_+^{\mathcal{O}^c}P_+^{\mathcal{O}} {\tau_t\pa{  X_{I_0}  }}P_+^{\mathcal{O}}P_+^{\mathcal{O}^c}= P_0 {\tau_t\pa{  X_{I_0}  }} P_0 \\ \
& =P_0 X P_0 = P_0 X^{+,+} P_0=0,
\end{align}
where we used \eq{tildeZ}, $P_+^{\mathcal{O}} P_+^{\mathcal{O}^c}=P_0$ and $X^{+,+} =0$.
Since $\widetilde{\tau_t\pa{  X_{I_0}  }}$ is supported on $\mathcal{O}^c$, we conclude that
$\pa{\widetilde{\tau_t\pa{  X_{I_0}  }}}^{+,+}=0$.    Thus we only need to estimate
$\pa{ P_-^{{\mathcal O}\setminus\mathcal T}P_+^{\pa{\mathcal T}}Y^{a,b}}_{I}$, where
$Y= \widetilde{\tau_t\pa{  X_{I_0}  }} $ and  $a,b=\pm $, but either $a=-$ or $b=-$.   
If $a=-$,  we have
 \begin{align}
P_-^{{\mathcal O}\setminus\mathcal T}P_+^{\pa{\mathcal T}}Y^{-,b}= P_-^{{\mathcal O}\setminus\mathcal T}P_+^{\pa{\mathcal T}}   P_-^{\mathcal{O}^c}  Y^{-,b}= P_-^{{\mathcal O}\setminus\mathcal T}   P_-^{\mathcal{O}^c}    P_+^{\pa{\mathcal T}}   Y^{-,b},
 \end{align}
and hence
\begin{align}
\E \pa{\sup_{t\in \R}\norm{\pa{P_-^{{\mathcal O}\setminus\mathcal T}P_+^{\pa{\mathcal T}}Y^{-,b}}_I}_1} &\le \norm{Y} \E \pa{\norm{P_I P_-^{{\mathcal O}\setminus\mathcal T}P_-^{\mathcal{O}^c}}_1}   \nn \\ & \le C \norm{X} \e^{- {\frac{1}{4}} m\ell},
\end{align}
using  \eqref{P-gP-2}. Since the $b=-$ case is similar we conclude from \eq{replaceOT1},\eq{replaceOT2}, and \eq{eq:X_ell} that
\beq\label{compproof3333}
\E \pa{\sup_{t\in \R}\norm{\pa{P_+^{\pa{\mathcal O}} \widetilde{\tau_t\pa{  X_{I_0}  }} - X_\ell(t)}_{I_0}}_1 }\le C \norm{X} \e^{- {\frac{1}{4}} m\ell}
\eeq

Combining \eq{compproof2111} and \eq{compproof3333} we get \eq{qlocI}.
\end{proof}

\section{Zero-velocity Lieb-Robinson bounds}\label{secLR}

In this section we prove Theorem~\ref{thm:expclusteringgenJ}.

 \begin{proof}[Proof of Theorem~\ref{thm:expclusteringgenJ}]
In view of \eq{X++0}, we can assume $X^{+,+}=Y^{+,+}=0$, and prove  the theorem in this case. This is the only step where we use cancellations from the commutator.  The estimate \eq{eq:dynloc} then follows immediately from Lemma~\ref{lemPXtgY}.

To prove \eq{eq:LRquasimix},  recall $P_{I_0}= P_I + P_0$,  and note that since $X^{+,+}=Y^{+,+}=0$ we have 
$P_0 X P_0= P_0 Y P_0=0$, so
\begin{align}\nn
&\left[ \tau_t\pa{X_{I_0}},Y_{I_0}\right]
 = \left[ \tau_t \pa{X_I},Y_I\right] + P_I\pa{ \tau_t\pa{X}P_0 Y- Y P_0  \tau_t\pa{X} }P_I
 \\   \nn & \quad  +P_0\pa{ \tau_t \pa{X} P_I Y - Y P_I \tau_t\pa{X} }P_0
+\pa{\tau_t\pa{X_I} Y P_0- P_0Y  \tau_t\pa{X_I} }\\  &\quad 
 +\pa{P_0 X \e^{-itH}Y_I - Y_I \e^{itH} XP_0} .\label{withcom}
\end{align}
Note that  $\left[ \tau_t\pa{X_I},Y_I\right] $ can be estimated  by \eq{eq:dynloc}.  We have
\begin{align}\nn
\norm{P_0 \tau_t\pa{X}P_I YP_0}_1 &= \norm{P_0 \tau_t\pa{X^{+,-}}P_I Y^{-,+}P_0}_1\\
& \le \norm{X}\norm{Y}\norm{ P_{-} ^{\pa{X}} \e^{-itH}P_I  P_{-} ^{\pa{Y}}}_1,\label{XJYP145}
\end{align}
so it can be estimated by \eq{P-gP-2}, with a similar estimate for $\norm{P_0 YP_I  \tau_t\pa{X}P_0}_1$.  Moreover,
\begin{align}\label{XJYP1}
&\norm{\tau_t\pa{X_I} Y P_0}_1=  \norm{\tau_t\pa{X_I} Y^{-,+} P_0}_1\\ \nn & \quad \le 
 \norm{X}\norm{Y} \norm{P_{-} ^{\pa{X}} \e^{-itH}P_IP_{-} ^{\pa{Y}}  }_1 + \norm{Y}\norm{P_I  X^{-,+} \e^{-itH}P_IP_{-} ^{\pa{Y}}  }_1.
\end{align}
The first term can be be estimated by \eq{P-gP-2}. To estimate the second term, let
$\ell= \dist(X,Y)\ge 1$. Then 
\begin{align}\nn
&\norm{P_I  X^{-,+} \e^{-itH}P_IP_{-} ^{\pa{Y}}  }_1\\ \nn
& \quad \le \norm{P_I  X^{-,+} P_{+} \up{\cS_{Y,\frac \ell 2}} \e^{-itH}P_I P_{-} ^{\pa{Y}}  }_1+ \norm{P_I  X^{-,+}  P_{-} \up{\cS_{Y,\frac \ell 2}}\e^{-itH}P_I P_{-} ^{\pa{Y}}  }_1
\\ & \quad \le \norm{X}\pa{\norm{ P_{+} \up{\cS_{Y,\frac \ell 2}} \e^{-itH}P_I P_{-} ^{\pa{Y}}  }_1+
\norm{P_I   P_{-} \up{\cS_{Y,\frac \ell 2}}  P_{-} ^{\pa{X}} }_1},
\end{align}
where we used $[X^{-,+},P_{-} \up{\cS_{Y,\frac \ell 2}} ]=0$.  Thus the second term in last line of \eq{XJYP1}  can be estimated by  \eq{P-gP-3} and \eq{P-gP-2}.

 The remaining three terms in \eq{withcom} can be similarly estimated.
(Although \eq{withcom} is stated for the commutator, it could have been stated separately for each term of the commutator. The above argument does not use cancellations from the commutator.) Combining all these estimates we get \eq{eq:LRquasimix}.  

It remains to prove \eq{eq:LRquasimix2}. Let  $X, Y$ and $Z$ be local observables.  In view of
 \eq{eq:LRquasimix}, we only need to estimate
 \beq
\E\pa{\sup_{t,s \in \R} \norm{ [P_I \pa{\tau_t\pa{ X}P_0 \tau_s(Y) - \tau_s(Y)P_0 \tau_t\pa{ X }}P_I,  Z_I ]}_1}.
 \eeq
 If we expand the commutator, we get  to estimate several terms, the first one being
\beq
 \E\pa{\sup_{t,s \in \R} \norm{ P_I \tau_t\pa{ X}P_0 \tau_s(Y)P_I Z P_I }_1}\le  \E\pa{\sup_{s \in \R} \norm{P_0 \tau_s(Y)P_IZ P_I }_1}
 \eeq
 This can be estimated as in \eq{XJYP145} and  \eq{XJYP1}, and the other terms can be similarly estimated, yielding 
\eq{eq:LRquasimix2}.

We will now  show  that for   the random   XXZ spin chain  the estimate \eq{eq:LRquasimix}    is not true  without the counterterms.   In fact, a stronger statement holds.   Let now $H$ be a random XXZ
 spin chain, and assume  that for all local observables $X$ and $Y$ we have 
 \begin{align}
  \E \pa{\sup_{t\in \R}\norm{\left[ \tau_t\pa{X_{I_0}},Y_{I_0}\right]}_1} \le C \|X\| \|Y\|  \Ups\pa{\dist (X,Y)},  
 \label{eq:LRquasimix44} 
 \end{align}
uniformly in $L$, where  the function $\Ups:\N \to [0,\infty)$ satisfies
$\lim_{r\to \infty} \Ups\pa{r}=0$.  Assume  \eq{eq:LRquasimix} holds with the same right hand side as  \eq{eq:LRquasimix44}.

It follows from \eq{eq:LRquasimix} and \eq{eq:LRquasimix44}
that
\begin{align}\nn 
&\E \pa{\norm{ \pa{{ X}P_0Y - YP_0 { X }}_I}_1}\le
\E \pa{\sup_{t\in \R}\norm{ \pa{\tau_t\pa{ X}P_0Y - YP_0 \tau_t\pa{ X }}_I}_1}\\ &\nn \hskip20pt \le 
 \E \pa{\sup_{t\in \R}\norm{\left[ \tau_t\pa{X_{I_0}},Y_{I_0}\right]- \pa{\tau_t\pa{ X}P_0Y - YP_0 \tau_t\pa{ X }}_I}_1} \\ \nn & \hskip110pt +  \E \pa{\sup_{t\in \R}\norm{\left[ \tau_t\pa{X_{I_0}},Y_{I_0}\right]}_1}\\ &\hskip20pt \le 
 2C \|X\| \|Y\|  \Ups\pa{\dist (X,Y)}\label{counter4}
\end{align}
In particular, taking $X=\sigma_i^x$ and $Y=\sigma_j^x$ we get (putting $L$ back in the notation)
\begin{align}\label{counter44}
\E \pa{\norm{ \pa{ \sigma_i^x P_0\up{L} \sigma_j^x - \sigma_j^x P_0\up{L}  \sigma_i^x }_I}_1}   \le 
 2C\,   \Ups\pa{\abs{i-j}}.
\end{align}

Thus, using  $\norm{A}_2^2\le \norm{A} \norm{A}_1 $ and  $\norm{\sigma_i^x P_0\up{L} \sigma_j^x - \sigma_j^x P_0\up{L}  \sigma_i^x }\le 2$, we get
\begin{align}\nn  
\E \pa{\norm{ \pa{ \sigma_i^x P_0\up{L} \sigma_j^x - \sigma_j^x P_0\up{L}  \sigma_i^x }_I}^2_2} & \le 2 \E \pa{\norm{ \pa{ \sigma_i^x P_0\up{L} \sigma_j^x - \sigma_j^x P_0\up{L}  \sigma_i^x }_I}_1} \\ & \le
4 C\,   \Ups\pa{\abs{i-j}}.\label{cont13}
\end{align}

Since \eq{cont13} is not compatible with \eq{eq:2terms0}, we have a contradiction, so 
\eq{eq:LRquasimix44}  cannot hold.
 \end{proof}

\section{General dynamical clustering}\label{secdyncl}

We now turn to the proof of  Theorem~ \ref{thmexpclust}.  
We will use the following lemma. 

\begin{lemma}\label{lem:filter} Let  $ \Theta_2 < \Theta_1  $.
 Given  $\alpha\in(0,1)$, there exist constants $m_\alpha>0$ and $C_\alpha<\infty$, such that, given $\Theta_3 \ge  \Theta_1$,  there exists a function  $f \in C_c^\infty(\R)$, such that 
\begin{enumerate}
\item $0\le f\le1$;
\item $\supp f\subset[\Theta_2,\Theta_3 + \Theta_1  - \Theta_2]$;
\item $f(x)=1$ for  $x\in[\Theta_1,\Theta_3 ]$;
\item $\abs{\hat f(t)}\le C_\alpha \e^{-m_\alpha \abs{t}^\alpha}$ for $\abs{t}\ge1$;
\item $\norm{\hat f}_1\le C_\alpha   \max\set{1,\ln \pa{\Theta_3-\Theta_2} }$.

\end{enumerate}
\end{lemma}

\begin{proof}Let $\theta=  \Theta_1 - \Theta_2$.  Pick a Gevrey class function $h\ge 0$ such that 
 \[ \supp h \subset [0,\theta];\; \int_\R h(x)\, \d x=1; \sqtx{and}  \abs{\hat h(t)}\le C_h \e^{-m_h \abs{t}^{\alpha}}\sqtx{for all} t\in\R,\] 
 where $C_h$ and $m_h >0$ are constants.
Let
 \[k(x)=\int_{-\infty}^x h\pa{y}\d y \qtx{for} x\in \R,\]
then $k\in C^\infty(\R)$ is non-decreasing and satisfies
\[0\le k \le 1,\quad \supp k\subset [0,\infty),\qtx{and} k(x)=1 \mbox{ for } x\ge \theta.\]

Given  $\Theta_3\ge  \Theta_1$, we claim that the function 
\beq
f(x)=k(x-\Theta_2)-k(x-\Theta_3 )
\eeq
has all the required properties.  Indeed, properties (i)--(iii) are obvious.  To finish, we compute
\beq
\hat f(t)=\int_\R \e^{-itx} \pa{\int^{x-\Theta_2}_{x-\Theta_3 }h(y)\,\d y} \d x.
\eeq
Integrating by parts and noticing that the boundary terms vanish, we get
\begin{align}\nn
\hat f(t)&=\tfrac{-i}{t}\int_\R \e^{-itx} \pa{h\pa{x-\Theta_2}-h\pa{x-\Theta_3 }} dx = \tfrac{-i}{t}\pa{\e^{-i\Theta_2t}-\e^{-i\Theta_3t}}\hat h(t)\\ & 
= \tfrac{-i}{t} \e^{-i\Theta_2t} \pa{1-\e^{-i(\Theta_3-\Theta_2) t}}\hat h(t).
\end{align}
Thus
\beq
\abs{\hat f(t)}\le 2 C_h\abs{\tfrac{\sin\pa{\frac 12\pa{ {\Theta_3-\Theta_2}}t}}{t}}\e^{-m_h\abs{t}^\alpha}\qtx{for all } t\in \R.
\eeq
Parts (iv) and (v)  follow. 
\end{proof}

We are ready to prove Theorem~\ref{thmexpclust}.

\begin{proof}[Proof of Theorem~\ref{thmexpclust}]  
Let  $H_\omega$  be a disordered  XXZ spin chain satisfying Property DL.  Let
$K= [\Theta_0, \Theta_2]$, where  $ \Theta_0 < \Theta_2 <\min\set{2 \Theta_0, \Theta_1 }$.
Since  \eq{eq:efcorDL} holds for the interval $[\Theta_0,\min\set{2 \Theta_0, \Theta_1 } ]$,   we assume $\Theta_1\le 2 \Theta_0$  without loss of generality.
We set  $K^\pr= (\Theta_2, \infty)$. 

 Let $X$ and $Y$ be local observables. In view of \eq{X++0}, we can assume $X^{+,+}=Y^{+,+}=0$, and prove  the theorem in this case. For a fixed $L$ (we omit $L$ from the notation), we have
 \begin{align}\label{RK1}
R_{K}  (\tau^K_t\pa{ X },Y)= \pa{\tau^K_t\pa{ X }\bar P_K Y}_K= \pa{\tau^K_t\pa{ X }P_{K^\pr} Y}_K+\pa{\tau^K_t\pa{ X }P_0 Y}_K.
 \end{align}

 Fix $\alpha\in(0,1)$, let    $\Theta_3 \ge 2 \Theta_2 $, to be chosen later,  and  let $f$ be the function given in Lemma ~\ref{lem:filter}.  We have
\begin{align}\label{tXK'Y}
\pa{\tau^K_t\pa{ X }P_{K^\pr} Y}_K = \pa{\tau^K_t\pa{ X }\pa{P_{K^\pr} -f(H)}Y}_K +\pa{\tau^K_t\pa{ X }f(H)Y}_K.
\end{align}

To estimate the first term, note that $P_{K^\pr }-f(H)=g(H)$, where $\abs{g}\le 1$  and    $g(H)= g(H) P_I  + g(H) \bar P\up{\Theta_3} $, where $P\up{\Theta_3}=P_{(-\infty, \Theta_3]}$ and  $\bar P\up{\Theta_3} = 1- P\up{\Theta_3}$.  The term with $g(H) P_I $ can be estimated by  Lemma~\ref{lemPXtgY},
\begin{align}\nn
\E\pa{\sup_{t\in \R}\norm {\pa{\tau^K_t\pa{ X }g(H)P_IY}_K}}&\le \E\pa{\sup_{t\in \R}\norm {\pa{\tau_t\pa{ X }g(H)P_IY}_I}} \\ & \le C\norm{X}\norm{Y}  \e^{-\frac 1 8 m\dist (X,Y)} .\label{RK2}
\end{align}
The contribution of  $g(H) \bar P\up{\Theta_3}$ is estimated by Lemma \ref{lem:Kitaev},
\begin{align}\nn
&\norm{\pa{\tau^K_t\pa{ X }g(H) \bar P\up{\Theta_3}Y}_K}\le  \norm{Y}\norm{P_K X g(H) \bar P\up{\Theta_3}}\\ & \nn \qquad \le  \norm{Y}\norm{P\up{\Theta_2} X g(H) \bar P\up{\Theta_3}}\le  \norm{Y}\norm{P\up{\Theta_2}X  \bar P\up{\Theta_3}}\\ & \qquad \le
 C_F \norm{X} \norm{Y} \e^{-\frac {m_F} {\abs{\cS_X}}\pa{\Theta_3-\Theta_2}}.\label{RK3}
\end{align}

To estimate  the second term on the right hand side of \eq{tXK'Y}, we recall that $H_K=0$ on  $ \supp f$, so 
\begin{align}
\pa{\tau^K_t\pa{ X }f(H)Y}_K=\e^{itH_K} \pa{Xf(H)Y}_K.
\end{align}
it follows from Lemmas~\ref{lemHast}, \ref{leminsertKf}  and \ref{lem:filter}  that 
 \begin{align}\label{RK4}
\pa{\tau^K_t\pa{ X }f(H)Y}_K=  A + T(K_f),
 \end{align}
 where
 \begin{align}\label{RK5}
 \norm{ A}\le 2 C_1  C_\alpha \norm{X} \norm{Y} \max\set{1,\ln \pa{\Theta_3-\Theta_2} }  \e^{- m_1\pa{ \dist (X,Y)}^\alpha},
 \end{align}
 \begin{align}
 T(J) = \e^{itH_K}  \pa{ \int_\R   \e^{-irH} YP_{J} \tau_r\pa{ X} \hat f(r) \, \d r }_K \qtx{for} J\subset\R,
 \end{align}
and 
 \beq
[0,2\Theta_2 - \Theta_1] \subset K_f \subset  [2\Theta_0  -\Theta_3-\pa{\Theta_1-\Theta_2}, 2\Theta_2-\Theta_2 ]\subset   (-\infty, \Theta_2 ].
 \eeq 
In view of \eq{gapcond},  $ P_{K_f}= P_{K_f^\pr} + P_0$, where $K_f^\pr=K_f \cap K$,
 so $T(K_f)=T(K_f^\pr) +T(\set{0})$.  We have
\begin{align}\nn
\E\pa{\sup_{r\in \R}\norm{T({K_f^\pr})}    }&\le \norm{\hat f}_1 \E\pa{\sup_{r\in \R} \norm{\pa{YP_{{K_f^\pr}} \tau_r\pa{ X} }_K}    } \\ & 
\le  C   \max\set{1,\ln \pa{\Theta_3-\Theta_2} }\norm{X}\norm{Y}  \e^{-\frac 1 8 m\dist (X,Y)},\label{RK6}
\end{align} 
where we used Lemmas~\ref{lemPXtgY} and \ref{lem:filter}.  In addition,
\begin{align}
T(\set{0})= \e^{itH_K}\pa{YP_{0}  X }_K = \pa{\tau_t^K\pa{Y}P_{0}  X }_K.\label{RK7}
\end{align}
To see this, let  $E,E^\pr \in K$. Proceeding as in \eq{P0mag},we have
\begin{align}\nn
& P_E \pa{ \int_\R   \e^{-irH} YP_0 \tau_r\pa{ X} \hat f(r) \, \d r }  P_{E^\pr}=\int_\R  P_E \e^{-irH} Y P_0 { X} \e^{-irH}P_{E^\pr}\hat f(r) \, \d r\\ \nn  & \quad = \pa{ \int_\R  \e^{-ir(E +E^\pr)}\hat f(r) \, \d r} P_E YP_0 { X} P_{E^\pr}= f(E+E^\pr )  P_E Y  P_0 { X} P_{E^\pr} \\  & \quad =
 P_E Y  P_0 { X} P_{E^\pr},
\label{P0mag3}
\end{align}
since $f(E+E^\pr)=1$ as $E+E^\pr \in [2\Theta_0,2\Theta_2]\subset [\Theta_1,\Theta_3]$.

Combining  \eq{RK1}, \eq{tXK'Y}, \eq{RK2}, \eq{RK3}, \eq{RK4},\eq{RK5},\eq{RK6}, and \eq{RK7}, we obtain
\begin{align}
&\norm{R_{K}  (\tau^K_t\pa{ X },Y) - \pa{\tau^K_t\pa{ X }P_0 Y}_K-  \pa{\tau_t^K\pa{Y}P_{0}  X }_K}\nn  \\ & \quad
\le C\norm{X}\norm{Y}\pa{ \max\set{1,\ln \pa{\Theta_3-\Theta_2} }  \e^{- m_2\pa{ \dist (X,Y)}^\alpha} +\e^{-\frac {m_F} {\abs{\cS_X}}\pa{\Theta_3-\Theta_2}}},
\end{align}
where $m_2= \min\set{m_1,\frac 1 8 m}>0$.

We now  choose $\Theta_3= \Theta_2 +  {\abs{\cS_X}}\pa{\dist (X,Y)}^\alpha$, note that $\Theta_3 \ge 2 \Theta_2$ if $\dist (X,Y)\ge \Theta_2^{\frac 1\alpha}$, obtaining
\begin{align}
&\norm{R_{K}  (\tau^K_t\pa{ X },Y) - \pa{\tau^K_t\pa{ X }P_0 Y}_K-  \pa{\tau_t^K\pa{Y}P_{0}  X }_K}\nn  \\ & \qquad
\qquad \qquad\le C\norm{X}\norm{Y}\pa {1+ \ln {\abs{\cS_X}} }  \e^{- m_3\pa{ \dist (X,Y)}^\alpha},
\end{align}
with  $m_3=\frac 12  \min\set{m_2,m_F}>0$, for $\dist (X,Y)$ sufficiently large.   
Observing that the argument can be done with $Y$ instead of $X$, we get \eq{eq:expclusteringgen'}.

Since
\begin{align}
\pa{[\tau^K_t\pa{ X },Y]}_K= R_K\pa{ \tau^K_t\pa{ X },Y}- R_K\pa{Y, \tau^K_t\pa{ X }}+ [\tau_t\pa{ X_K },Y_K],
\end{align}
 \eq{eq:dynloc3333} follows immediately from \eq{eq:expclusteringgen'} and \eq{eq:dynloc}.

To conclude the proof, we need to show  that for   a random   XXZ spin chain $H$  the estimates \eq{eq:expclusteringgen'} and   \eq{eq:dynloc3333}  are not true  without the counterterms.   

Suppose \eq{eq:expclusteringgen'} holds without counterterms, even in a weaker form: for all local observables $X$ and $Y$ we have
\begin{align}\nn  
& \E \pa{\sup_{t\in \R}\norm{R_K\pa{ \tau^K_t\pa{ X },Y} }}\\
& \qquad \qquad \le   C\pa{{\min\set{\abs{\cS_{X}},\abs{\cS_{Y}}}}} \|X\| \|Y\|  \Ups\pa{\dist (X,Y)},
\label{eq:expclusteringgen3}
\end{align}
uniformly in $L$, where  the function $\Ups:\N \to [0,\infty)$ satisfies
$\lim_{r\to \infty} \Ups\pa{r}=0$.  Assume  \eq{eq:expclusteringgen'} holds with the same right hand side as  \eq{eq:expclusteringgen3}.  Taking $X=\sigma_i^x$ and $Y=\sigma_j^x$, 
and proceeding   as in \eq{counter4}-\eq{counter44}, we get  (putting $L$ back in the notation)
\begin{align}
&\E \pa{\norm{ \pa{ \sigma_i^x P_0\up{L} \sigma_j^x +\sigma_j^x P_0\up{L}  \sigma_i^x }_K}} \le
4 C\,   \Ups\pa{\abs{i-j}}.\label{cont139}
\end{align}

Recall that (in the notation of the proof of Lemma~\ref{lem:spillterms}, as in \eq{Tnotation}),
\begin{align}
& Z: =  \pa{ \sigma_i^x P_0\up{L} \sigma_j^x +\sigma_j^x P_0\up{L}  \sigma_i^x }_K = 
T\pa{P_K\up{L} \delta_i, P_K\up{L} \delta_j} +   T\pa{ P_K\up{L} \delta_j,P_K\up{L} \delta_i}.
\end{align}
Let $V$ be the two dimensional vector space spanned by the vectors $P_K\up{L} \delta_i$ and $P_K\up{L} \delta_j$, and let $Q_V$ be the orthogonal projection onto $V$. We clearly  have
$Z= Q_V Z Q_V$ and  $\norm{Z}\le 2$, and hence 
\begin{align}
\norm{Z}_2^2 \le 2 \norm{Z}^2  \le 4 \norm{Z} ,
\end{align}
so it follows from \eq{eq:2terms0} that  there exist constants $\gamma_K >0$  and $R_K$ such that
\begin{align}\nn
\E \pa{\norm{ \pa{ \sigma_i^x P_0\up{L} \sigma_j^x +\sigma_j^x P_0\up{L}  \sigma_i^x }_K}} &\ge \tfrac 1 4
\E \pa{\norm{ \pa{ \sigma_i^x P_0\up{L} \sigma_j^x +\sigma_j^x P_0\up{L}  \sigma_i^x }_K}_2^2}\\ & \ge   \tfrac 1 4 \gamma_K,  \label{cont13977}
\end{align}
for all  $i,j \in \Z$ with $\abs{i-j}\ge R_K$.

  Since  \eq{cont139} and \eq{cont13977}   establish a contradiction, we conclude that \eq{eq:expclusteringgen3} cannot hold.

We show the necessity of the counterterms in \eq{eq:dynloc3333} in a similar way.   Note that the counterterm  for $X=\sigma_i^x$ and $Y=\sigma_j^x$  is given by $ Z\up{L}(t)$ as in \eq{Zt}.    If we  assumed the validity of  \eq{eq:dynloc3333} without counterms, we would have
\beq\label{cont13945}
\E \pa{\sup_{t\in \R} \norm{ Z\up{L}(t)}} \le 4 C\,   \Ups\pa{\abs{i-j}},
\eeq
where the function  $\Ups$ is as in \eq{cont139}.  Since $Z\up{L}(t)$ is a rank $4$ operator,
we have
\begin{align}
\norm{Z\up{L}(t)}_2^2 \le 4 \norm{Z}^2  \le 16 \norm{Z} ,
\end{align}
and hence
\begin{align}
\sup_{t\in \R} \norm{ Z\up{L}(t)} \ge  \tfrac 1 {16}\sup_{t\in \R} \norm{ Z\up{L}(t)}_2^2\ge  \tfrac 1 {16} \lim_{T\to \infty} \tfrac 1 T \int_0^T \norm{Z\up{L}(t)}_2^2\, \d t,
\end{align}
so \eq{cont13945} and \eq{eq:4terms1} give a contradiction, an hence \eq{cont13945} cannot hold.
\end{proof}

\appendix

\section{A priori  transition probabilities for spin chains}\label{appspinc}
The following lemma is 
 an adaptation of 
 \cite[Lemma~6.6(2)]{Arad} to our needs. It holds for every spin chain with uniformly norm-bounded next-neighbor  interactions (more generally,  for uniformly norm-bounded interactions of fixed finite range).

 Given a spin chain Hamiltonian $H\up{L} $ and an energy  $E\in \R$,we write
 $P\up{E,L}= \chi_{(-\infty,E]}(H\up{L})$ for the Fermi projection, and let
 $\bar P\up{E,L}=1-P\up{E,L}$.

\begin{lemma}\label{lem:Kitaev}   Let $
H\up{L}= \sum_{i=-L}^{L-1} Y_{i,i+1}\up{L} $ be a spin chain Hamiltonian on $\cH\up{L}=\otimes_{i\in [-L,L]}\C_i^2$,   $\C_i^2=\C^2$ for  $i\in \Z$, where $Y_{i,i+1}\up{L} $ is a local observable with support  $\cS_{Y_{i,i+1}\up{L}}= [i,i+1]$ for $i\in [-L,L-1]$. Suppose
\beq
\max_{i\in [-L,L-1]} \norm{ Y_{i,i+1}\up{L}} \le \theta <\infty.
\eeq
Then there exists constants  $m_F>0$ and  $C_F<\infty $, depending  on  $\theta$, but  and independent of $L$, such that for any local  observable $X$  and energies  $E< E^\pr $ we have 
\beq
\norm{ P\up{E,L}X \bar P\up{E^\pr,L}} \le C_F \norm{X}   \e^{-\frac {m_F} {\abs{\cS_X}}\pa{E'-E}}.
\eeq 
\end{lemma}

Note that if $H=H_\bom$ is a disordered XXZ spin chain,  we can write  (cf. \eq{finiteXXZ})
\beq
H_\bom\up{L}= \sum_{i=-L}^{L-1} Y_{\omega; i,i+1}\up{L} ,
\eeq
where  $Y_{\omega; i,i+1}\up{L} $ is a local observable with  $\cS_{Y_{\omega; i,i+1}\up{L}}= [i,i+1]$, and
\beq
\sup_{\omega \in [0,1]^{\Z}, L\in \N}\max_{i\in [-L,L-1]} \norm{ Y_{\omega; i,i+1}\up{L}} \le \theta= \frac 1 2\pa{1+\tfrac 1 \Delta} +2 \lambda + \beta <\infty.
\eeq

\begin{proof}[Proof of Lemma~\ref{lem:Kitaev}] Let $E< E^\pr $ and let $X$ be a  local  observable.  Without loss of generality we take  $\norm{X}=1$. We proceed as in  \cite[Proof of Lemma~6.6(2)]{Arad}.  For all $r>0$ we have (we omit $L$ from the notation)
\begin{align}\label{Kit2}
\norm{  \bar P\up{E^\pr}X P\up{E}}\le \e^{-r\pa{E^\pr -E}}\norm{\e^{rH} X \e^{-rH}}
\end{align}
Hadamard's Lemma gives
\begin{align}
\e^{rH} X \e^{-rH}= X + \sum_{n=1}^\infty \tfrac {r^n}{n!} ad_H^{\, n} (X), \qtx{where} ad_H ( \cdot)= [H,\cdot],
\end{align}
and hence
\begin{align}
\norm{\e^{rH} X \e^{-rH}} \le 1 +\sum_{n=1}^\infty \tfrac {r^n}{n!} \norm{ad_H^{\, n} (X)}.
\end{align}
Letting $\cS=\cS_X= [s_X,r_X]$,
$\gamma= \abs{\cS_X}=r_X-s_X +1$, and 
$\cS_{j,k}= [s_X-j,r_X+k]\cap [-L,L]$ for $j,k=0,1,2,\ldots$, we can see that
\begin{align}
ad_H (X)= [H,X] = \sum_{j=1}^{J_1} Z\up{1}_j, \qtx{with} J_1\le \gamma+1,
\end{align}
where each  $ Z\up{1}_j=  [Y_{i,i+1} ,X]$ for some $i\in \cS_{1,0}$, so $\norm{Z\up{1}_j}\le 2 \theta$ and either  $\cS_{Z\up{1}_j}\subset  \cS_{1,0}$ or $\cS_{Z\up{1}_j}\subset  \cS_{0,1}$,  and we have 
\begin{align}
\norm{ad_H (X)}\le 2 \theta (\gamma +1).
\end{align}
Using induction, we can show that for $n=1,2,3,\ldots$ (with $J_0=1$)
\begin{align}
ad_H^{\, n} (X)=  \sum_{j=1}^{J_n} Z\up{n}_j,  \qtx{with} J_n\le (\gamma +n) J_{n-1},
\end{align}
where each  $Z\up{n}_j$ is a local observable with   $\norm{Z\up{n}_j}\le \pa{2 \theta}^n$ and   $\cS_{Z\up{n}_j}\subset  \cS_{k,k^\pr}$  for some $k,k^\pr \in \set{0,1,2\ldots}$ with $k+k^\pr \le n$, and we have 
\begin{align}\nn
\norm{ad_H^{\, n} (X)}&\le \pa{2 \theta}^n J_n \le \pa{2 \theta}^n \prod_{k=1}^n (\gamma + k)
\le \pa{2 \theta \gamma}^n \prod_{k=1}^n (1 + k)\\ &= \pa{2 \theta \gamma}^n (n+1)!
\end{align}

We conclude that
\begin{align}
\norm{\e^{rH} X \e^{-rH}}\le  C_r= 1+ \sum_{n=1}^\infty \pa{2 r \theta \gamma}^n (n+1).
\end{align}
Choosing $r= \pa{4 \theta \gamma}^{-1}$, we get $\tilde C= C_{\pa{4 \theta \gamma}^{-1}} <\infty$,
and it follows from \eq{Kit2} that
\beq
\norm{  \bar P\up{E^\pr}X P\up{E}}\le \tilde C  \e^{-\pa{4 \theta \gamma}^{-1}\pa{E^\pr -E}},
\eeq
proving the lemma.
\end{proof}


\end{document}